\documentclass[titlepage,12pt]{article}
\usepackage{fullpage}
\usepackage{amsmath,amssymb,amsthm,bm,graphicx,natbib,subcaption,multirow,hyperref,tikz}
\usepackage{algorithm,algorithmic,multirow,rotating,placeins,subcaption}
\usepackage{colortbl}
\usetikzlibrary{shapes}
\newtheorem{theorem}{Theorem}
\newtheorem{lem}{Lemma}

\usepackage{adjustbox}
\usepackage{array}
\newcolumntype{R}[2]{%
    >{\adjustbox{angle=#1,lap=\width-(#2)}\bgroup}%
    l%
    <{\egroup}%
}
\newcommand*\rot{\multicolumn{1}{R{45}{1em}}}

\begin{document}

\title{Quantifying the Multi-Scale Performance of Network Inference Algorithms\footnote{Running title: Multi-Scale Scores for Network Inference}}
\author{Chris. J. Oates\thanks{Corresponding author: \texttt{c.oates@warwick.ac.uk}} and Simon E. F. Spencer \\ Department of Statistics, University of Warwick, Coventry, UK \and Richard Amos \\ The MathWorks, Cambridge, UK.}

\maketitle

\begin{abstract}
Graphical models are widely used to study complex multivariate biological systems.
Network inference algorithms aim to reverse-engineer such models from noisy experimental data.
It is common to assess such algorithms using techniques from classifier analysis.
These metrics, based on ability to correctly infer individual edges, possess a number of appealing features including invariance to rank-preserving transformation.
However, regulation in biological systems occurs on multiple scales and existing metrics do not take into account the correctness of higher-order network structure.

In this paper novel performance scores are presented that share the appealing properties of existing scores, whilst capturing ability to uncover regulation on multiple scales. 
Theoretical results confirm that performance of a network inference algorithm depends crucially on the scale at which inferences are to be made; in particular strong local performance does not guarantee accurate reconstruction of higher-order topology.
Applying these scores to a large corpus of data from the DREAM5 challenge, we undertake a data-driven assessment of estimator performance.
We find that the ``wisdom of crowds'' network, that demonstrated superior local performance in the DREAM5 challenge, is also among the best performing methodologies for inference of regulation on multiple length scales.

MATLAB R2013b code \verb+net_assess+ is provided as Supplement.

\end{abstract}

\section{Introduction}

Graphical representations of complex multivariate systems are increasingly prevalent within systems biology. In general a graph or \emph{network} $G = (V,E)$ is characterised by a set $V$ of vertices (typically associated with molecular species) and a set $E \subseteq V \times V$ of edges, whose interpretation will be context-specific.
In many situations the edge set or \emph{topology} $E \equiv E(G)$ is taken to imply conditional independence relationships between species in $V$ \citep{Pearl}.
For fixed and known vertex set $V$, the data-driven characterisation of network topology is commonly referred to as \emph{network inference}.

In the last decade many approaches to network inference have been proposed and exploited for several different purposes \citep{Oates}.
In some settings it is desired to infer single edges with high precision \citep[e.g.][]{Hill}, whereas in other applications it is desired to infer global connectivity, such as subnetworks and clusters \citep[e.g.][]{Breitkreutz}.
In cellular signalling systems, the scientific goal is often to identify a set of upstream regulators for a given target, each of which is a candidate for therapeutic intervention designed to modulate activity of the target \citep{Morrison,Winter}.
The output of network inference algorithms are increasingly used to inform the design of experiments \citep{Nelander,Hill} and may soon enter into the design of clinical trials \citep{Chuang,Heiser}.
It is therefore important to establish which network inference algorithms work best for each of these distinct scientific goals. 

Assessment of algorithm performance can be achieved {\it in silico} by comparing inferred networks to known data-generating networks. It can also be achieved using data obtained {\it in vitro}; however this requires that the underlying biology is either known by design \citep{Cantone}, well-characterised by interventional experiments \citep{Maathuis2}, or estimated from larger corpora of data \citep{Weile}. In either case an estimated network $\hat{G}$, typically represented as a weighted adjacency matrix, is compared against a known or assumed benchmark network $G$.
Community-wide blind testing of network inference algorithms is performed at the regular DREAM challenges \citep[see \url{http://www.the-dream-project.org/};][]{Marbach2012,Prill}. 

There is broad agreement in the network inference literature regarding the selection of suitable performance scores (described below), facilitating the comparison of often disparate methodologies across publications.
In this literature, the quality of an estimated network $\hat{G}$ with respect to a benchmark $G$ is assessed using techniques from classifier analysis. 
That is, each possible edge $(i,j) \in V \times V$ has an associated class label $Z(i,j) = \mathbb{I}\{(i,j) \in E(G)\}$, where $\mathbb{I}$ is the indicator function.
A network estimator $\hat{G}$ may then be seen as an attempt to estimate $Z(i,j)$ for each pair $(i,j)$.
Two of the main performance scores from classifier analysis are area under the receiver operating characteristic curve (AUROC) and area under the precision-recall curve (AUPR), though alternative performance scores for classification also exist (e.g. \cite{Drummond}).
These scores, formally defined in Sec.~\ref{classification}, are based on \emph{confusion matrices} of true/false positive/negative counts and represent essentially the only way to quantify performance at a local level (i.e. based on individual edges).
At present, performance assessment in the network inference literature does not typically distinguish between the various scientific purposes for which network inference algorithms are to be used.
Yet network inference algorithms are now frequently employed to perform diverse tasks, including identifying single edges with high precision \citep{Hill}, eliciting network motifs such as cliques \citep{Wang,Feizi} or learning a coherent global topology such as connected components \citep{Breitkreutz}.

Whilst performance for local (i.e. edge-by-edge) recovery is now standardised, there has been comparatively little attention afforded to performance scores that capture ability to recover higher-order features such as cliques, motifs and connectivity.
Recent studies, including \cite{Banerjee,Jurman,Jurman2}, proposed to employ spectral distances as a basis for comparing between two networks on multiple length scales.
In this article we present several additional multi-scale scores (MSSs) for network inference algorithms, each of which reflects ability to solve a particular class of inference problem.
Much of the popularity of existing scores derives from their objectivity, interpretability and invariance to rank-preserving transformation. 
Unlike previous studies, we restrict attention only to MSS that satisfy these desiderata.

The remainder of this paper proceeds as follows:
In Section \ref{methods} we formulate the assessment problem, state our desiderata and present novel performance scores that satisfy these requirements whilst capturing aspects of network reconstruction on multiple scales.
Using a large corpus of estimated and benchmark networks from the DREAM5 Challenge in Section \ref{results}, we survey estimator performance and conduct an objective, data-driven examination of the statistical power of each MSS.
The proposed MSSs provide evidence that the ``wisdom of crowds'' approach, that demonstrated superior (local) performance in the DREAM5 challenge, also offers gains on multiple length scales.
Sections \ref{discussion} and \ref{conclude} provide a discussion of our proposals and suggest directions for future work.
MATLAB R2013b code \verb+net_assess+ is provided in the Supplement, to accelerate the dissemination of ideas discussed herein.

\section{Methods} \label{methods}

We proceed as follows: Sections \ref{assumptions} and \ref{desiderata} clarify the context of the assessment problem for network inference algorithms amd  list certain desiderata that have contributed to the popularity of local scores.
Sections \ref{notation} and \ref{classification} introduce graph-theoretic notation and review standard performance assessment based on recovery of individual edges.
In Sections \ref{mss one} and \ref{mss two} we introduce several novel MSSs for assessment of network inference algorithms. 
We require each MSS to satisfy our list of desiderata; however these scores differ from existing scores by assessing inferred network structure on several (in fact all) scales. 
For each MSS we discuss associated theoretical and computational issues.
Finally Section \ref{test} describes computation of $p$-values for the proposed MSSs.

\subsection{Problem Specification} \label{assumptions}

Performance assessment for network inference algorithms may be achieved by comparing estimated networks against known benchmark information.
The interpretation of the estimated networks themselves has often been confused in existing literature, with no distinction draw between the contrasting notions of significance and effect size.
In this Section we therefore formally state our assumptions on the interpretation of both the benchmark network $G$ and the network estimators or estimates $\hat{G}$.

\begin{itemize}
\item [A1] All networks are directed, unsigned and contain no self-edges. 
\end{itemize}
A network is {\it signed} if each edge carries an associated $+/-$ symbol.
(A1) is widely applicable since an undirected edge may be recast as two directed edges and both signs and self-edges may simply be removed. 
The challenge of inferring signed networks and more generally the problem of predicting interventional effects requires alternative performance scores that are not dealt with in this contribution, but are surveyed briefly in Sec.~\ref{discussion}.
The preclusion of self-edges aids presentation but it not required by our methodology.

The form of this benchmark information will influence the choice of performance score and we therefore restrict attention to the most commonly encountered scenario:
\begin{itemize}
\item [A2] The benchmark network $G$ is unweighted.
\end{itemize}
A network is {\it weighted} if each edge has an associated weight $w \in \mathbb{R}$.
Note that the case of unweighted benchmark networks is widely applicable, since weights may simply be removed if necessary.
We will write $\mathcal{G}_0$ for the space of all directed, unweighted networks that do not contain self-edges and write $\mathcal{G}$ for the corresponding space of directed, weighted networks that do not contain self-edges.
\begin{itemize}
\item [A3] The benchmark network $G$ contains at least one edge and at least one non-edge. 
\end{itemize}
\begin{itemize}
\item [A4] Network estimators $\hat{G}$ are weighted ($\hat{G} \in \mathcal{G}$), with weights having the interpretation that larger values indicate a larger (marginal) probability of the corresponding edge being present in the benchmark network.
\end{itemize}
In particular we do not consider weights that instead correspond to effect size (see Sec.~\ref{discussion}). 
\begin{itemize}
\item [A5] In all networks, edges refer to a \emph{direct} dependence of the child on the parent at the level of the vertex set $V$; that is, not occurring via any other species in $V$. 
\end{itemize}
Assumptions (A1-5) are typical for comparative assessment challenges such as DREAM \citep{Marbach,Prill}.

\subsection{Performance Score Desiderata} \label{desiderata}

Fix a benchmark network $G$.
A {\it performance score} is defined as function $S:\mathcal{G} \times \mathcal{G}_0 \rightarrow \mathbb{R}$ that accepts an estimated network $\hat{G} \in \mathcal{G}$ and a benchmark network $G \in \mathcal{G}_0$ and returns a real value $S(\hat{G},G)$ that summarises some aspect of $\hat{G}$ with respect to $G$.
Examples of performance scores are given below.
Our approach revolves around certain desiderata that any (i.e. not just multi-scale) performance score $S$ ought to satisfy:
\begin{itemize}
\item [D0] (Interpretability) $S(\hat{G},G) \in [0,1]$ for all $\hat{G} \in \mathcal{G}$, $G \in \mathcal{G}_0$, with larger values corresponding to better performance at some specified aspect of network reconstruction.
\item [D1] (Computability) $S$ should be readily computable.
\item [D2] (Objectivity) $S$ should contain no user-specified parameters.
\end{itemize}
A network estimate $\hat{G} \in \mathcal{G}$ is called $S$-{\it optimal} for a benchmark network $G$ if it maximises the performance score $S(\cdot,G)$ over all networks in $\mathcal{G}$.
\begin{itemize}
\item [D3] (Consistency) The oracle estimator $\hat{G} = G$ is $S$-optimal.
\end{itemize}
A \emph{rank-preserving} transformation $\{x_i\} \mapsto \{x_i'\}$ of a collection of real values $x_i \in \mathbb{R}$ satisfies $x_i' < x_j'$ whenever $x_i < x_j$.
\begin{itemize}
\item [D4] (Invariance) $S$ should be invariant to rank-preserving transformations of the weights associated with an estimate $\hat{G}$. 
\end{itemize}
The criteria (D0-3) are important for practical reasons; (D4) is more technical and reflects the fact that we wish to compare estimators whose weights need not belong to the same metric space. i.e. different algorithms may be compared on the same footing, irrespective of the actual interpretation of edge weights (subject to (A4)).
Given that much of the popularity of standard classifier scores derives from (D0-4), it is important that any proposed MSS also satisfies the above desiderata.
As we will see below, previous studies, such as \cite{Banerjee,Jurman,Jurman2,Peters}, do not satisfy (D2,4), precluding their use in objective assessment such as the DREAM Challenges.
In Secs. \ref{mss one}, \ref{mss two} below we present novel MSS that satisfy each of (D1-4).

\subsubsection{Examples of Performance Scores} \label{related}

The local performance scores that have become a standard in the literature, defined in Sec.~\ref{classification} below, are easily shown to satisfy the above desiderata.
To date, the challenge of assessing network inference algorithms over multiple scales has received little statistical attention; yet there are several general proposals for quantifying higher order network structure. Notably \cite{Jurman2} assessed spectral distances for suitability in application to biological networks. The authors recommended use of the Ipsen-Mikhailov distance (IMD) due to its perceived stability and robustness properties, as quantified on randomly generated and experimentally obtained networks \citep{Jurman}.
IMD evaluates the difference of the distribution of Laplacian eigenvalues between two networks and therefore could be considered an MSS.
There has, to date, been no examination of IMD and related spectral metrics in the context of performance assessment for network inference applications.
However IMD and related constructions \citep[e.g.][]{Banerjee} fail to satisfy the above desiderata: (i) IMD itself fails to satisfy (D4) since eigenvalues are not invariant to rank-preserving transformations; (ii) the IMD formula contains a user-specified parameter, invalidating (D2).

\cite{Peters} recently introduced ``Structural Intervention Distance'' (SID; actually a pre-distance) which interprets inferred networks $\hat{G}^{\tau}$ as estimators of an underlying causal graph $G$, where both $\hat{G}^{\tau}$ and $G$ must be acyclic \citep{Pearl}.
Specifically, SID is the count of pairs of vertices $(i,j)$ for which the estimate $\hat{G}^{\tau}$ incorrectly predicts intervention distributions (in the sense of \cite{Pearl}) within the class of distributions that are Markov with respect to $G$.
As such, SID is closely related to the challenge of predicting the effect of unseen interventions.
However SID fails to satisfy (D0) due to the requirement that estimated and benchmark networks must be acyclic (i.e. a causal interpretation will not be justified in general) and is therefore unsuitable for this application.

\subsection{Notation} \label{notation}

Below we introduce notation that will be required to define our proposed scores:
Throughout this paper an unweighted network $G \in \mathcal{G}_0$ on vertices $V = \{1,2,\dots,p\}$ (where $p < \infty$) is treated as a binary matrix $G \in \{0,1\}^{p \times p}$ with $(i,j)$th entry denoted $G(i,j)$.
When $G(i,j) = 1$ or $G(i,j) = 0$ we say that $G$ contains or does not contain the edge $(i,j)$, respectively.
Write $G(\bullet,j) = \{i \in V : G(i,j) = 1\}$ for the set of parents for vertex $j \in V$.
In this paper we do not allow self-edges (A1), so that $G(i,i) = 0$ for all $i \in V$.
A path $P$ from $i$ to $j$ in $G \in \{0,1\}^{p \times p}$ is characterised by a sequence of vertices 
\begin{eqnarray}
P = (p_0,p_1,\dots,p_{m-1},p_m) \in V^{m+1}, \; \; \; p_0 = i, \; p_m = j
\end{eqnarray}
such that $G(p_{k-1},p_k) = 1$ and $p_k \neq i$ for all $k \geq 1$. 
Note that a path $P$ may contain cycles, but these cycles may not involve $p_0$.
We say that this path has length $\ell(P) = m$ where $1 \leq m < \infty$. Let $G(i \rightarrow j)$ denote the set of all paths in $G$ from $i$ to $j$, so that in particular $G(i \rightarrow i) = \emptyset$.

We identify weighted, unsigned networks $\hat{G} \in \mathcal{G}$ with non-negative real-valued matrices $\hat{G} \in [0,1]^{p \times p}$ and denote the $(i,j)$th entry by $\hat{G}(i,j)$.
Should a network inference procedure produce edge weights in $[0,\infty)$ then by dividing through by the largest weight produces a network with weights in $[0,1]$. This highlights the importance of property (D4), invariance to rank preserving transformations.
It will be required to map $\hat{G}$ into the space of unweighted networks by thresholding the weights at a certain level $\tau$. 
For $0 \leq \tau \leq 1$ we write $\hat{G}^{\tau}$ for the unweighted network corresponding to a matrix with entries $\hat{G}^{\tau}(i,j) = \mathbb{I} \{\hat{G}(i,j) \geq \tau\}$.

\subsection{Local Scores} \label{classification}

In this Section we briefly review local scoring, which has become an established standard in the network inference literature.

\noindent{\bf Definition:} Classification performance scores are defined as functions of confusion matrices, that count the number of true positive (TP), false positive (FP), true negative (TN) and false negative (FN) calls produced by a classifier. 
From (A4) it follows that the $k$ largest entries in a network estimate $\hat{G}$ correspond to the $k$ pairs of vertices with largest marginal probabilities of being present as edges in the benchmark network $G$; it is therefore reasonable to threshold entries of $\hat{G}$, say at a level $\tau$, and to consider confusion matrices corresponding to the classifier $\hat{G}^{\tau}$.
In the standard case of {\it local} estimation, confusion matrices are given by TP$(\tau) = \sum_{i,j} \hat{G}^{\tau}(i,j)G(i,j)$, FP$(\tau) = \sum_{i,j} \hat{G}^{\tau}(i,j)(1-G(i,j))$, TN$(\tau) = \sum_{i,j} (1-\hat{G}^{\tau}(i,j))(1-G(i,j))$ and FN$(\tau) = \sum_{i,j} (1-\hat{G}^{\tau}(i,j))G(i,j)$.
Based on these quantities, (local) performance scores may be defined.
In particular, one widely used score is the area under the receiver operating characteristic (ROC) curve
\begin{eqnarray}
S_{\text{local}}^{\text{ROC}} = \int \text{TPR} \; d\text{FPR}
\end{eqnarray} 
where $\text{TPR} = \text{TP}/(\text{TP}+ \text{FN})$ is the {\it true positive rate} and $\text{FPR} = \text{FP}/(\text{FP}+ \text{TN})$ is the {\it false positive rate}.
$S_{\text{local}}^{\text{ROC}}$ has an interpretation as the probability that a randomly selected pair from $\{(i,j) \; : \; G(i,j) = 1\}$ is assigned a higher weight $\hat{G}(i,j)$ that a randomly selected pair from $\{(i,j) \; : \; G(i,j) = 0\}$ \citep{Fawcett}. As such $S_{\text{local}}^{\text{ROC}}$ takes values in $[0,1]$ with 1 representing perfect performance and $1/2$ representing performance that is no better than chance.
For finitely many test samples, $S_{\text{local}}^{ROC}$ curves may be estimated by linear interpolation of points $(\text{FPR}(\tau),\text{TPR}(\tau))$ in ROC-space. 

Precision-recall (PR) curves are an alternative to ROC curves that are useful in situations where the underlying class distribution is skewed, for example when the number of negative examples greatly exceeds the number of positive examples \citep{Davis}. 
For biological networks that exhibit sparsity, including gene regulatory networks \citep{Tong} and metabolic networks \citep{Jeong}, the number of positive examples (i.e. edges) is frequently smaller than the number of negative examples (i.e. non-edges).
In this case performance is summarised by the area under the PR curve
\begin{eqnarray}
S_{\text{local}}^{\text{PR}} = \int \text{PPV} \; d\text{TPR} ,
\end{eqnarray}
where $\text{PPV} = \text{TP}/(\text{TP}+\text{FP})$ is the {\it positive predictive value}.
(In this paper we adopt the convention that $\text{PPV} = 0$ whenever $\text{TP}+\text{FP} = 0$.)
$S_{\text{local}}^{\text{PR}}$ also takes values in $[0,1]$, with larger values representing better performance.
For finitely many test samples in PR-space, unlike ROC-space, linear interpolation leads to over-optimistic assessment of performance. The MATLAB code provided in the Supplement follows \cite{Goadrich} in using nonlinear interpolation to achieve correct, unbiased estimation of area under the PR curve.

\noindent{\bf Captures:} Local scores capture the ability of estimators to recover the exact placement of individual edges. Both $S_{\text{local}}^{\text{ROC}}$ and $S_{\text{local}}^{\text{PR}}$ summarise the ability to correctly infer local topology across a range of thresholds $\tau$, with PR curves prioritising the recovery of positive examples in the ``top $k$ edges'' \citep{Fawcett}.

\noindent{\bf Desiderata:} (D1; computability), (D2; objectivity) and (D3; consistency) are clearly satisfied.  
(D4; invariance) is satisfied, since the image of a parametric curve is invariant to any monotone transformation of the parameter (in this case $\tau$).

\subsection{Multi-Scale Score 1 (MSS1)} \label{mss one}

We now introduce the first of our MSSs, which targets ability to infer the connected components of the benchmark network $G$, through the existence or otherwise of directed paths between vertices.

\noindent{\bf Definition:} In the notation of Sec.~\ref{notation}, local performance scores are based on estimation of class labels $Z(i,j) = G(i,j)$ associated with individual edges.
A natural generalisation of this approach is to assign labels $Z(i,j) \in \{0,1\}$ to pairs of vertices $(i,j) \in V \times V$, such that $i \neq j$, based on {\it descendancy}; that is, based on the presence or otherwise of a directed path from vertex $i$ to vertex $j$ in the benchmark network. i.e. $Z(i,j) = \mathbb{I}\{G(i \rightarrow j) \neq \emptyset\}$.
By comparing descendancies in $\hat{G}^{\tau}$ against descendancies in $G$, we can compute confusion matrices and, by analogy with local scores, we can construct area-under-the-curve scores by allowing the threshold $\tau$ to vary. We denote these scores respectively as $S_{\text{MSS1}}^{\text{ROC}}$ and $S_{\text{MSS1}}^{\text{PR}}$.
By analogy with local scores, both $S_{\text{MSS1}}^{\text{ROC}}$, $S_{\text{MSS1}}^{\text{PR}}$ take values in $[0,1]$, and $S_{\text{MSS1}}^{\text{ROC}}$ is characterised as the probability that a randomly selected pair from $\{(i,j) \; : \; G(i \rightarrow j) \neq \emptyset\}$ is assigned a higher weight 
\begin{eqnarray}
\tau_{i,j} = \max_{P \in \hat{G}^{\tau}(i \rightarrow j)} \min_{0 \leq i \leq \ell(P)} \hat{G}(p_i,p_{i+1})
\end{eqnarray}
than a randomly selected pair from $\{(i,j) \; : \; G(i \rightarrow j) = \emptyset\}$.

\noindent{\bf Captures:} MSS1 captures the ability to identify ancestors and descendants of any given vertex. MSS1 scores therefore capture the ability of estimators to recover connected components on all length scales.

\noindent{\bf Desiderata:} (D1; computability) is satisfied due to the well-known Warshall algorithm for finding the transitive closure of a directed, unweighted network \citep{Warshall}, with computational complexity $\mathcal{O}(p^3)$.
For the estimated network $\hat{G}$ it is required to know, for each pair $(i,j)$ the value of $\tau_{i,j}$, i.e. the largest value of the threshold $\tau$ at which $\hat{G}^{\tau}(i \rightarrow j) \neq \emptyset$.
To compute these quantities we generalised the Warshall algorithm to the case of weighted networks; see Sec.~\ref{theory results}.
(D2-4) are automatically satisfied analogously to local scores, by construction of confusion matrices and area under the curve statistics.

\subsection{Multi-Scale Score 2 (MSS2)} \label{mss two}

Whilst MSS1 scores capture the ability to infer connected components, they do not capture the graph theoretical notion of {\it differential connectivity} (i.e. the minimum number of edges that must be removed to eliminate all paths between a given pair of vertices).
Our second proposed MSS represents an attempt to explicitly prioritise pairs of vertices which are highly connected over those pairs with are weakly connected:

\noindent{\bf Definition:} For each pair $(i,j) \in V \times V$ we will compute an \emph{effect} $0 \leq e_{ij} \leq 1$ that can be thought of as the importance of variable $i$ on the regulation of variable $j$ according to the network $G$ (in a global sense that includes indirect regulation).
To achieve this we take inspiration from recent work by \cite{Feiglin} as well as \cite{Morrison, Winter}, who exploit spectral techniques from network theory.
Since the effect $e_{ij}$, which is defined below, includes contributions from all possible paths from $i$ to $j$ in the network $\hat{G}$, it explicitly captures differential connectivity. 

We formally define effects for an arbitrary unweighted network $H$ (which may be either $G$ or $\hat{G}^{\tau}$).
Specifically the effect $e_{ij}$ of $i$ on $j$ is defined as the sum over paths
\begin{eqnarray}
e_{ij} = \delta_{ij} + \sum_{P\in H(i \rightarrow j)} \prod_{k=1}^{\ell(P)} \frac{1}{|H(\bullet,p_k)|} 
\label{effects def}
\end{eqnarray}
where $\delta_{ij}$ is the Kronecker delta.
Effects $e_{ij}$ quantify direct and indirect regulation on all length scales. To illustrate this, notice that Eqn.~\ref{effects def} satisfies the recursive property
\begin{eqnarray}
e_{ij} = \frac{1}{|H(\bullet,j)|} \sum_{k \in H(\bullet,j)} e_{ik}
\label{recurse}
\end{eqnarray}
(see Sec.~\ref{theory results}).
Intuitively, Eqn.~\ref{recurse} states that the fraction of $j$'s behaviour explained by $i$ is related to the combined fractions of $j$'s parents' behaviour that are explained by $i$.
Rephrasing, in order to explain the behaviour of $j$ it is sufficient to explain the behaviour of each of $j$'s parents.
Moreover, if a parent $k$ of $j$ is an important regulator (in the sense that $j$ has only a small number of parents) then the effect $e_{ik}$ of $i$ on $k$ will contribute significantly to the combined effect $e_{ij}$.
Eqn.~\ref{recurse} is inspired by \cite{Page} and later \cite{Morrison, Winter}, but differs from these works in two important respects: (i) For biological networks it is more intuitive to consider normalisation over incoming edges rather than outgoing edges, since some molecular species may be more influential than others.
Mathematically, \cite{Page} corresponds to replacing $|H(\bullet,j)|$ in Eqn.~\ref{recurse} with $|H(i,\bullet)|$.
(ii) \cite{Page} employed a ``damping factor'' that imposed a multiplicative penalty on longer paths, with the consequence that effects were readily proved to exist and be well-defined. In contrast our proposal does not include damping on longer paths (see discussion of D2 below) and the theory of \cite{Page} and others does not directly apply in this setting.

Below we write $\bm{e} = \{e_{ij}\}$ for the matrix that collects together all effects for the benchmark network $G$; similarly denote by $\hat{\bm{e}}^{\tau}$ the matrix of effects for the thresholded estimator $\hat{G}^{\tau}$.
Any well-behaved measure of similarity between $\hat{\bm{e}}^{\tau}$ and $\bm{e}$ may be used to define a MSS. 
We constructed an analogue of a confusion matrix as TP$= \sum_{i,j} \hat{e}_{ij}^{\tau}e_{ij}$, FP$=\sum_{i,j} \hat{e}_{ij}^{\tau}(1-e_{ij})$, TN$=\sum_{i,j} (1-\hat{e}_{ij}^{\tau})(1-e_{ij})$, FN$=\sum_{i,j} (1-\hat{e}_{ij}^{\tau})e_{ij}$.
Repeating the construction across varying threshold $\tau$, we compute analogues of ROC and PR curves. 
(Note that, unlike conventional ROC curves, the curves associated with MSS2 need not be monotone; see Supp. Sec.~\ref{AUC con}.)
Finally scores $S_{\text{MSS2}}^{\text{ROC}}$, $S_{\text{MSS2}}^{\text{PR}}$ are computed as the area under these curves respectively.

\noindent{\bf Captures:}
MSS2 is a spectral method, where larger scores indicate that the inferred network $\hat{G}$ better captures the {\it eigenflows} of the benchmark network $G$ \citep{Lakhina}.
In general neither $S_{\text{MSS2}}^{\text{ROC}}$ nor $S_{\text{MSS2}}^{\text{PR}}$ need have a unique maximiser (see Sec.~\ref{theory results}). As such, MSS2 scores do not require precise placement of edges, provided that higher-order topology correctly captures differential connectivity. 

\noindent{\bf Desiderata:}
Eqn.~\ref{effects def} involves an intractable summation over paths: Nevertheless (D1; computability) is ensured by an efficient iterative algorithm related but non-identical to \cite{Page}, described in Sec.~\ref{theory results}.
In order to ensure (D2; objectivity) we did not include a damping factor that penalised longer paths, since the amount of damping would necessarily depend on the nature of the data and the scientific context.
It is important, therefore, to establish whether effects $e_{ij}$ are mathematically well defined in this objective limit.
This paper contributes novel mathematical theory to justify the use of MSS2 scores and prove the correctness of the associated algorithm (see Sec.~\ref{theory results}).
As with MSS1, the remaining desiderata (D3-4) are satisfied by construction.

\subsection{Significance Levels} \label{test}

In performance assessment of network inference algorithms we wish to test the null hypothesis $H_0$ : $\hat{G} \sim \mathcal{M}_0$ for an appropriate null model $\mathcal{M}_0$.
We will construct a one-sided test based on a performance score $S$ and we reject the null hypothesis $H_0$ when $S(\hat{G},G) \in (s^*,1]$ for an appropriately chosen critical value $s^*$.
In general the distribution of a score $S(R,G)$ under a network-valued random variable $R$ displays a nontrivial dependence on the benchmark network $G$. 
We therefore follow \cite{Marbach2012} and propose a Monte Carlo approach to compute significance levels, though alternative approaches exist including \cite{Scutari}. 
Specifically, significance of an inferred network $\hat{G}$ under a null model $\mathcal{M}_0$ is captured by the $p$-value $p(\hat{G}) := \mathbb{P}(S(R,G) \geq S(\hat{G},G))$ and given by the strong law of large numbers as
\begin{eqnarray}
\frac{1}{n} \sum_{i=1}^n \mathbb{I}(S(R_i,G) \geq S(\hat{G},G)) \overset{a.s.}{\rightarrow}  p(\hat{G}) \; \; \; \; \; \text{ as } n \rightarrow \infty
\end{eqnarray}
where the $R_i$ are independent samples from $\mathcal{M}_0$. 

The choice of null model $\mathcal{M}_0$ is critical to the calculation and interpretation of $p$-values $p(\hat{G})$. 
In most biological applications, the null would ideally encompass biological sample preparation, experimental data collection, data preprocessing and network estimation; however in practice it is convenient to define a null model $\mathcal{M}_0 \equiv \mathcal{M}_0(\hat{G})$ conditional upon the inferred network $\hat{G}$ \citep[e.g.][]{Marbach}.
For brevity, we restrict attention to the following choice: A random network $R = (\pi(V),E(\hat{G}))$ is obtained from $\hat{G}$ by applying a uniformly random permutation $\pi$ to the vertex labels $V$; that is, $R$ is a uniform sample from the space of graph isomorphisms of $\hat{G}$.
This choice has the interpretation that, for deterministic network inference algorithms, the null $\mathcal{M}_0(\hat{G})$ corresponds to randomly permuting the variable labels in the experimental dataset, prior to both data preprocessing and network inference.
Note that $\mathcal{M}_0(\hat{G})$ results in a trivial ($R_i$-independent) hypothesis test when $\hat{G}$ is either empty or complete; we do not consider these degenerate cases in this paper.

In Sec.~\ref{signif} we empirically assess the power of this test based on both existing and proposed scores $S$.
MATLAB R2013b code \verb+net_assess+, that was used to compute all of the scores and associated $p$-values used in this paper, is provided in the Supplement.

\section{Results} \label{results}

In Sec.~\ref{theory results} below we present theoretical results relating to the proposed scores, demonstrating that MSS and local scores can yield arbitrarily different conclusions in settings where local topology is recovered very well but higher-order topology is recovered very poorly and vice versa.
In Sec.~\ref{DREAM info}, in order to empirically assess the statistical power our proposed scores, we appeal to the large corpus of data available from the DREAM community, that represents a large and representative sample of network estimators used by the community.

\subsection{Theoretical Results} \label{theory results}

Whilst distributional results are difficult to obtain, we are able to characterise the behaviour of MSS in both favourable and unfavourable limits.
Our first result proves non-uniqueness of $S$-optimal estimates:

\begin{theorem}[Non-uniqueness of $S$-optimal networks] \label{nonunique}
For each score $S_{\text{MSS1}}^{\text{ROC}}$, $S_{\text{MSS1}}^{\text{PR}}$, $S_{\text{MSS2}}^{\text{ROC}}$ and $S_{\text{MSS2}}^{\text{PR}}$, the benchmark network $G$ is always $S$-optimal (D3). However, for MSS1 and MSS2, $S$-{\it optimal} estimates are not unique. 
\end{theorem}
\noindent This result contrasts with local performance scores, where $S_{\text{local}}^{\text{ROC}}$ and $S_{\text{local}}^{\text{PR}}$ are uniquely maximised by the benchmark network $G$. Intuitively, Theorem \ref{nonunique} demonstrates that it is possible to add or remove individual edges from a network without changing its higher-order features, such as its connected components.

Our second result proves that the MSS are fundamentally different from existing scores by demonstrating it is possible to simultaneously recover local topology with arbitrary accuracy yet fail to correctly accurately recover the higher-order topology (and vice versa):

\begin{theorem}[MSS are distinct from local scores] \label{different}
(a) For any $\epsilon > 0$ there exist a benchmark network $G \in \mathcal{G}_0$ and an estimate $\hat{G} \in \mathcal{G}$ that satisfy (i) $S_{\text{local}}^{\text{ROC}}$, $S_{\text{local}}^{\text{PR}} >1-\epsilon$, (ii) $S_{\text{MSS1}}^{\text{ROC}}$, $S_{\text{MSS2}}^{\text{ROC}} < \frac{1}{2} + \epsilon$, and (iii) $S_{\text{MSS1}}^{\text{PR}}$, $S_{\text{MSS2}}^{\text{PR}} < \epsilon$.

(b) For any $\epsilon > 0$ there exist a benchmark network $G \in \mathcal{G}_0$ and an estimate $\hat{G} \in \mathcal{G}$ that satisfy (i) $S_{\text{MSS1}}^{\text{ROC}}$, $S_{\text{MSS2}}^{\text{ROC}}$, $S_{\text{MSS1}}^{\text{PR}}$, $S_{\text{MSS2}}^{\text{PR}} >1-\epsilon$, (ii) $S_{\text{local}}^{\text{ROC}} < \frac{1}{2} + \epsilon$, and (iii) $S_{\text{local}}^{\text{PR}} < \epsilon$.
\end{theorem}
An important consequence of this result is that estimators that perform well locally, including DREAM Challenge winners, may not be appropriate for scientific enquiry regarding non-local network topology.
Given the increasing use of network inference algorithms in systems biology \citep[e.g.][]{Chuang,Breitkreutz}, this highlights the need to study performance of network inference algorithms on multiple length scales.

Our final theoretical contribution is to propose and justify efficient algorithms for computation of both MSS1 and MSS2.
For MSS1 we generalise the Warshall algorithm to compute the transitive closure of a directed, weighted graph:

\begin{theorem}[Computation of MSS1] \label{Warshall}
For a fixed weighted network $H$, let $\tau_{i,j} = \max\{\tau : H^{\tau}(i \rightarrow j) \neq \emptyset\}$ be the largest value of the threshold parameter $\tau$ for which there exists a path from $i$ to $j$ in the network $H^{\tau}$. Then $\bm{\tau} = \{\tau_{i,j}\}$ may be computed as follows:
\begin{algorithmic}    
\STATE $\bm{\tau} = H$              
\FOR{$k \leftarrow 1:p$}
\FOR{$i \leftarrow 1:p$}
\FOR{$j \leftarrow 1:p$}
\STATE $\tau_{i,j} \leftarrow \max(\tau_{i,j},\min(\tau_{i,k},\tau_{k,j}))$    
\ENDFOR
\ENDFOR
\ENDFOR
\RETURN $\bm{\tau}$
\end{algorithmic}
\end{theorem}

For MSS2, following \cite{Page} we formulate effects $e_{ij}$ as solutions to an eigenvalue problem. However our approach is non-identical to \cite{Page} (see Sec.~\ref{mss two}); as a consequence we must prove that (i) effects are mathematically well-defined, and (ii) an appropriately modified version of the algorithm of \cite{Page} converges to these effects.
Below we contribute the relevant theory:

\begin{theorem}[Computation of MSS2] \label{compute}
The effects $e_{i,j}$ for an unweighted network $H$ exist, are unique, satisfy Eqn.~\ref{recurse} and may be computed as
\begin{eqnarray}
\bm{e}_i = \lim_{n \rightarrow \infty} \bm{v}^{(n)}, \; \; \; \bm{v}^{(n)} = \frac{\bm{M}^{(i)}\bm{v}^{(n-1)}}{\|\bm{M}^{(i)}\bm{v}^{(n-1)}\|}, \; \; \; v^{(0)}_k = \delta_{ik}
\label{algo}
\end{eqnarray}
where $\bm{e}_i$ is the $i$th row of $\bm{e}$ and $\bm{M}^{(i)}$ is defined by (i) $M_{im}^{(i)} = \delta_{im}$, (ii) $M_{km}^{(i)} = \frac{1}{|H(\bullet,k)|}$ for all $m \in H(\bullet,k)$ and $H(i \rightarrow k) \neq \emptyset$, (iii) $M_{km}^{(i)} = 0$ otherwise.
\end{theorem}
\noindent Theorem \ref{compute} facilitates highly efficient numerical computation of the effect matrices $\bm{e}$ required for MSS2.
In practice Eqn.~\ref{algo} requires convergence diagnostics; for experiments in this paper we used the stopping rule $\|\bm{v}^{(n)}-\bm{v}^{(n-1)}\|_1 < 0.01 \times p$.

\subsection{Empirical Results} \label{DREAM info}

For an objective comparison covering a wide range of network inference methodologies, we exploited the results of the DREAM5 challenge, where 36 methodologies were blind-tested against both simulated and real transcription factor data \citep{Marbach2012}.
The network inference challenge included (i) {\it in silico} data generated using GeneNetWeaver \citep{Marbach2009}, (ii) experimentally validated interactions from a curated database for {\it E. coli}, and (iii) transcription-factor binding data and evolutionarily conserved binding motifs for {\it S. Cerevisiae}.
Data were provided on the genomic scale, involving (i) 1643, (ii) 4297 and (iii) 5667 genes respectively.
For each of these systems, benchmark information was provided in the form of targets for a subset of genes. 
For assessment purposes we restricted attention to predictions made regarding these subnetworks only, satisfying (A1-3), consisting of (i) 195, (ii) 334 and (iii) 333 genes respectively.
The DREAM5 data are well-suited for this investigation, representing a wide variety of network reconstruction algorithms contributed by the community and containing sufficiently many samples to provide a definitive assessment of statistical power.

In each challenge, participants were required to provide a list of inferred edges $(i,j)$, together with associated weights $0 \leq \hat{G}(i,j) \leq 1$.
Such edges were directed, unsigned and excluded self-edges, satisfying (A1).
Participants were asked to ``{\it infer the structure of the underlying transcriptional regulatory networks}'', so that edges were ranked in terms of their probability of existence, rather than their inferred effect size (A4). 
For DREAM5, (A5; direct dependence) appears to be implicit.
The algorithms used to obtain edge lists and weights varied considerably, being classified according to their statistical formulation as either {\it Bayesian}, {\it Correlation}, {\it Meta}, {\it Mutual Information} (MI), {\it Regression} or {\it Other}.
For each methodology we obtained performance scores for datasets (i-iii) using the \verb+net_assess+ package (Supp. Figs. \ref{silico bar}, \ref{ecoli bar}, \ref{sc bar}). 
We also investigated the Overall Score reported by \cite{Marbach2012}, that summarises performance by combining local AUROC and AUPR scores across all 3 datasets (further details in Supp. Sec.~\ref{overall def}).

In Section \ref{participant} below we examine the characteristics of our proposed scores in an aggregate sense based on the full DREAM5 data.
Section \ref{indiv} then considers estimator-specific score profiles in order to understand the multi-scale properties of different network inference algorithms.
In particular Section \ref{community} focusses on a ``wisdom of crowds'' approach that has previously demonstrated strong local performance in this setting.
Finally Section \ref{signif} quantifies the statistical power of our proposed scores in the context of hypothesis testing.

\subsubsection{Aggregate Results} \label{participant}

Local and MSS scores may be arbitrarily different in principle (Theorem \ref{different}), however in practice these scores may be highly correlated.
To investigate this we produced scatter plots between the existing and proposed performance scores over all 36 estimators and all 3 datasets (a total of $n = 108$ samples; Fig.~\ref{DREAM}, Supp. Fig.~\ref{DREAM2}).
Our findings, based on Spearman correlation (Table \ref{corr table}), showed that all scores were significantly correlated under AUPR on the {\it in silico} and {\it E. Coli} datasets ($p \leq 0.01$), but not on the {\it S. Cereviviae} dataset. Indeed results for the {\it S. Cerevisiae} dataset were less impressive for all methods, under all performance scores, in line with the general poor performance on this dataset reported by \cite{Marbach2012}. 
On the combined data, both MSS1 and MSS2 was significantly correlated with the local score under AUROC ($\rho = 0.84$, $\rho = 0.57$; $p < 10^{-10}$ in each case) suggesting that accurate reconstruction of local topology is associated with accurate reconstruction of higher-order topology.
On the other hand, on all 3 datasets and both AUROC and AUPR statistics, MSS2 was less highly correlated with the local score compared to MSS1 (e.g. on the combined data and AUPR, $\rho = 0.82$ versus $\rho = 0.29$).
This suggests that estimated local topology is more predictive of descendancy relationships than predictive of differential connectivity and spectral features.
In many cases MSS1 and MSS2 were significantly correlated; this was expected since both estimators target multiple scales by design.

\begin{figure}[t]
\centering
\includegraphics[width = \textwidth]{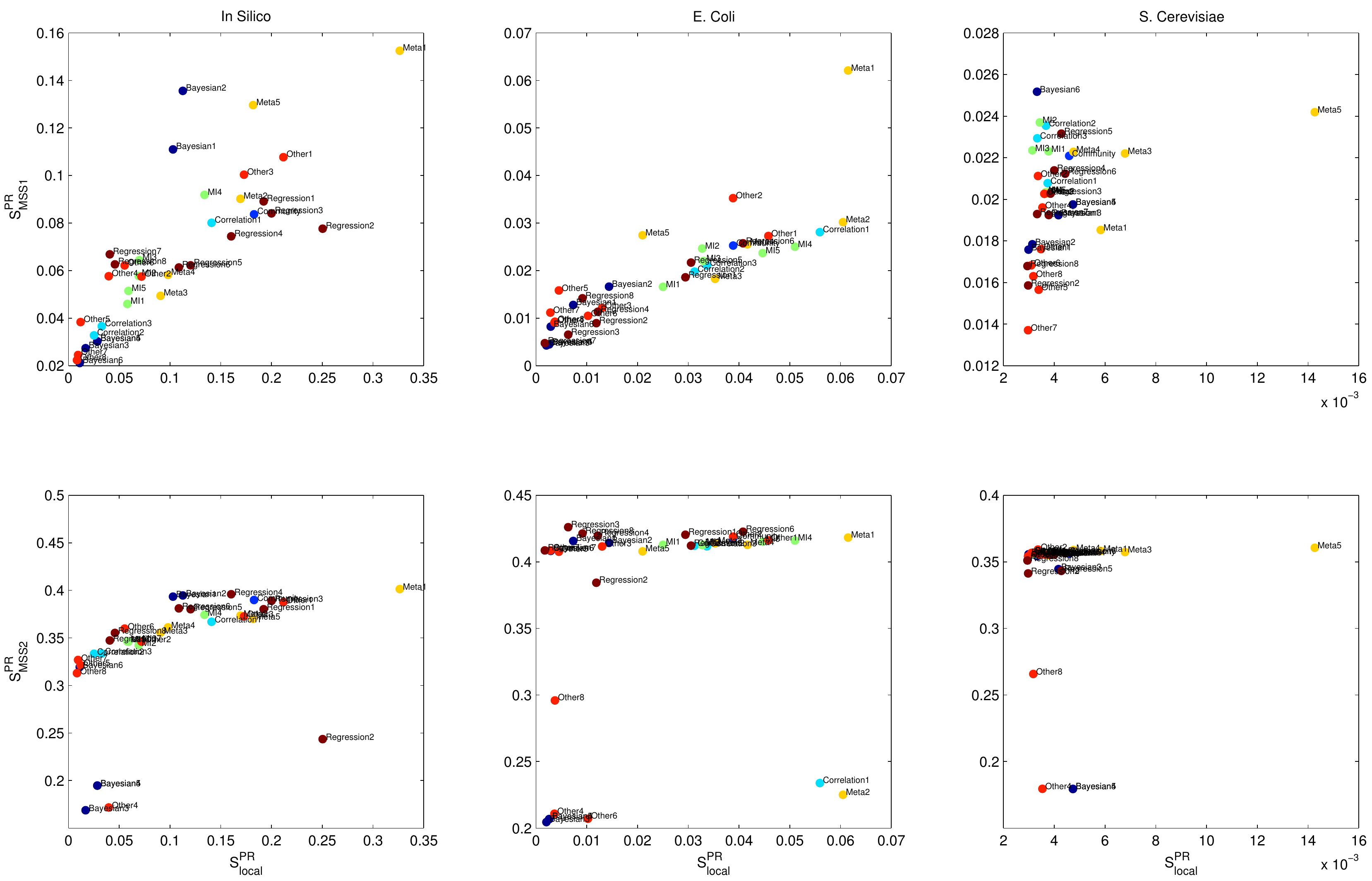}
\caption{DREAM5 network inference challenge data. [The 36 DREAM5 methodologies were assessed using both local and multi-scale performance scores. Participants were grouped according to their statistical formulation as either {\it Bayesian}, {\it Correlation}, {\it Meta}, {\it Mutual Information} (MI), {\it Regression} or {\it Other}.
{\it Community} represents a crowd-sourced network estimator proposed by \cite{Marbach2012}.
Left to right: {\it In silico} dataset, {\it E. Coli} dataset, {\it S. Cerevisiae} dataset.
Here we show results based on area under the precision recall curve.]}
\label{DREAM}
\end{figure}

\begin{table}[t]
\centering
\resizebox{0.9\textwidth}{!}{
\begin{tabular}{|c|cccc|cccc|l|}
\multicolumn{1}{r}{} & \rot{Local} & \rot{Marbach {\it et al.}} & \rot{MSS1} & \rot{MSS2} & \rot{Local} & \rot{Marbach {\it et al.}} & \rot{MSS1} & \rot{MSS2} & \multicolumn{1}{r}{} \\ \hline
& & 0.60 & 0.80 & 0.64 & & 0.78 & 0.87 & 0.77 & Local \\
{\it In Silico} & \cellcolor{black!25} $< 10^{-4}$ & & 0.53 & 0.78 & \cellcolor{black!25} $< 10^{-10}$ & & 0.53 & 0.52 & Marbach {\it et al.} \\
($n = 36$) & \cellcolor{black!25} $< 10^{-10}$ & \cellcolor{black!25} $< 10^{-3}$ & & 0.51 & \cellcolor{black!25} $< 10^{-20}$ & \cellcolor{black!25} $< 10^{-3}$ & & 0.82 & MSS1 \\
& \cellcolor{black!25} $< 10^{-4}$ & \cellcolor{black!25} $< 10^{-10}$ & \cellcolor{black!25} $< 10^{-2}$ & & \cellcolor{black!25} $< 10^{-10}$ & \cellcolor{black!25} $< 10^{-2}$ & \cellcolor{black!25} $< 10^{-10}$ & & MSS2 \\ \hline
& & 0.42 & 0.92 & 0.13 & & 0.58 & 0.92 & 0.44 & Local \\
{\it E. Coli} & \cellcolor{black!25} 0.01 & & 0.27 & 0.46 & \cellcolor{black!25} $< 10^{-3}$ & & 0.50 & 0.54 & Marbach {\it et al.} \\
($n = 36$) & \cellcolor{black!25} $< 10^{-20}$ & \cellcolor{black!25} 0.11 & & 0.05 & \cellcolor{black!25} $< 10^{-20}$ & \cellcolor{black!25} $< 10^{-2}$ & & 0.42 & MSS1 \\
& \cellcolor{black!25} 0.46 & \cellcolor{black!25} $< 10^{-2}$ & \cellcolor{black!25} 0.77 & & \cellcolor{black!25} $< 10^{-2}$ & \cellcolor{black!25} $< 10^{-3}$ & \cellcolor{black!25} 0.01 & & MSS2 \\ \hline
& & 0.26 & 0.70 & 0.17 & & 0.29 & 0.44 & 0.18 & Local \\
{\it S. Cerevisiae} & \cellcolor{black!25} 0.13 & & 0.13 & 0.45 & \cellcolor{black!25} 0.09 & & 0.22 & 0.23 & Marbach {\it et al.} \\
($n = 36$) & \cellcolor{black!25} $< 10^{-5}$ & \cellcolor{black!25} 0.46 & & -0.06 & \cellcolor{black!25} $< 10^{-2}$ & \cellcolor{black!25} 0.19 & & 0.32 & MSS1 \\
& \cellcolor{black!25} 0.31 & \cellcolor{black!25} $< 10^{-2}$ & \cellcolor{black!25} 0.74 & & \cellcolor{black!25} 0.29 & \cellcolor{black!25} 0.17 & \cellcolor{black!25} 0.06 & & MSS2 \\ \hline
& & 0.29 & 0.84 & 0.57 & & 0.37 & 0.82 & 0.29 & Local \\
Combined & \cellcolor{black!25} $< 10^{-2}$ & & 0.28 & 0.34 & \cellcolor{black!25} $< 10^{-4}$ & & 0.27 & 0.36 & Marbach {\it et al.} \\
($n = 108$) & \cellcolor{black!25} $< 10^{-20}$ & \cellcolor{black!25} $< 10^{-2}$ & & 0.40 & \cellcolor{black!25} $< 10^{-20}$ & \cellcolor{black!25} $< 10^{-2}$ & & 0.04 & MSS1 \\
& \cellcolor{black!25} $< 10^{-10}$ & \cellcolor{black!25} $< 10^{-3}$ & \cellcolor{black!25} $< 10^{-4}$ & & \cellcolor{black!25} $< 10^{-2}$ & \cellcolor{black!25} $< 10^{-3}$ & \cellcolor{black!25} 0.69 & & MSS2 \\ \hline
\multicolumn{1}{r}{} & \multicolumn{4}{c}{AUROC} & \multicolumn{4}{c}{AUPR} & \multicolumn{1}{r}{} \\ 
\end{tabular}
}
\caption{Spearman correlation between performance scores (white) and associated $p$-values (grey); DREAM5 network inference challenge data. [Performance scores include area under the receiver operating characteristic (AUROC) and precision recall (AUPR) curves, based on both local scores and the proposed Multi-Scale Scores (MSS). The Overall Score of \cite{Marbach2012} combines both local AUROC and AUPR scores as described in Supp. Sec.~\ref{overall def}.]}
\label{corr table}
\end{table}

\subsubsection{Individual Participant Performance} \label{indiv}

Network inference algorithms are predominantly employed in contexts where there is considerable uncertainty regarding the data-generating network topology. In such settings attention is often restricted to the ``most significant'' aspects of inferred topology; we therefore similarly restrict attention to individual estimator performance as quantified by scores based on AUPR, reserving AUROC results for Supp. Sec.~\ref{AUROC results}.

Under MSS1 the best performing methods included {\it Bayesian 1,2,6} and {\it Meta 1,5} (Fig.~\ref{DREAM}).
For MSS2 many methods attained similarly high scores; these included {\it Bayesian 1,2,6} and {\it Meta 1,5}. 
Interestingly these 5 best performing algorithms did not rank highly according to the (local) Overall Score of \cite{Marbach} (placing 8th, 24th, 20th, 22nd and 34th respectively out of 36).
Similarly {\it Regression 2}, which almost maximised $S_{\text{local}}^{\text{PR}}$ on the {\it in silico} dataset, was among algorithms with the lowest $S_{\text{MSS2}}^{\text{PR}}$.
This appears to be a real example of the conclusion of Theorem \ref{different}, where {\it Regression 2} recovers several individual edges yet fails faithfully reconstruct higher-order topology, including connected components.
Supp. Fig.~\ref{meta5} displays inferred and benchmark topology for {\it Regression 2}; it is visually clear that the algorithm fails to distinguish between different connected components in the true data-generating network. 
Conversely, estimators such as {\it Bayesian 1,2} simultaneously achieve strong multi-scale performance and unimpressive local performance (Fig.~\ref{DREAM}). 
It is well understood that Bayesian estimators are well suited to recovery of a coherent joint graphical model, compared to estimators that decompose network inference into independent neighbourhood selection problems \citep{Marbach}.
These results suggest that this intuition is manifest in the multi-scale performance of Bayesian network inference algorithms as quantified by both $S_{\text{MSS1}}^{\text{PR}}$ and $S_{\text{MSS2}}^{\text{PR}}$.

\subsubsection{Community Performance} \label{community}

A key finding of \cite{Marbach2012} was that the {\it Community} network, a principled aggregation of the DREAM5 participants' predictions, was able to maximise the Overall Score for local performance (the so-called ``wisdom of crowds'' phenomenon).
Interestingly, we found that this {\it Community} network also performed well over multiple length scales (Supp. Figs. \ref{silico bar},\ref{ecoli bar},\ref{sc bar}). 
In fact the {\it Community} network was almost $S_{\text{MSS2}}^{\text{ROC}}$ and $S_{\text{MSS2}}^{\text{PR}}$-optimal over all 36 methodologies on all 3 datasets and was ranked highly under $S_{\text{MSS1}}^{\text{ROC}}$ in all experiments.
These results suggest that aggregation of estimators that perform well locally (e.g. {\it Regression 2}) and well on higher-order topology (e.g. {\it Bayesian 2}) may be a successful strategy to achieve strong performance over multiple length scales.

\subsubsection{Significance Testing} \label{signif}

A common fault of complex test statistics is that they lack power relative to simpler statistics (e.g. \cite{Simon}).
In order to explore the relative statistical power of existing and proposed performance scores $S$ in hypothesis testing, we computed Monte Carlo $p$-values for each network estimate $\hat{G}$ in the DREAM5 challenge, taking each score in turn as a test function (see Sec.~\ref{test}).
Under the assumption that all DREAM5 participants were capable, in principle, of recovering some aspects of structural information (i.e. estimators were not simply random samples from $\mathcal{M}_0$), we may estimate statistical power by treating the DREAM5 data as cases where $H_0$ should be rejected.
Results (Fig.~\ref{pvals}) showed that, over all 36 participants and all 3 datasets, $S_{\text{local}}^{PR}$ assigned a significance level $p < 0.05$ to 88\% of the networks, $S_{\text{MSS1}}^{\text{PR}}$ achieved 73\% whilst $S_{\text{MSS2}}^{\text{PR}}$ achieved 71\% at this threshold.
Varying the threshold $1 > p > 0.01$ we see that $S_{\text{local}}^{\text{PR}}$ is roughly 25\% more powerful compared to both $S_{\text{MSS1}}^{\text{PR}}$ and $S_{\text{MSS2}}^{\text{PR}}$.
Given that challenge participants likely developed methodology to optimise local performance scores rather than multi-scale scores, these findings support the conclusion that the multi-scale scores $S_{\text{MSS1}}^{\text{PR}}$, $S_{\text{MSS2}}^{\text{PR}}$ offer comparable power relative to $S_{\text{local}}^{\text{PR}}$.
From the perspective of methodology, these results show that most DREAM5 algorithms are better suited to identifying local rather than higher-order topology.
Results for ROC-based scores demonstrated a more considerable decrease in statistical power between local and multi-scale test statistics; it may therefore be prudent to restrict attention to multi-scale results based on PR-scores.

\begin{figure}[t]
\centering
\includegraphics[width = \textwidth]{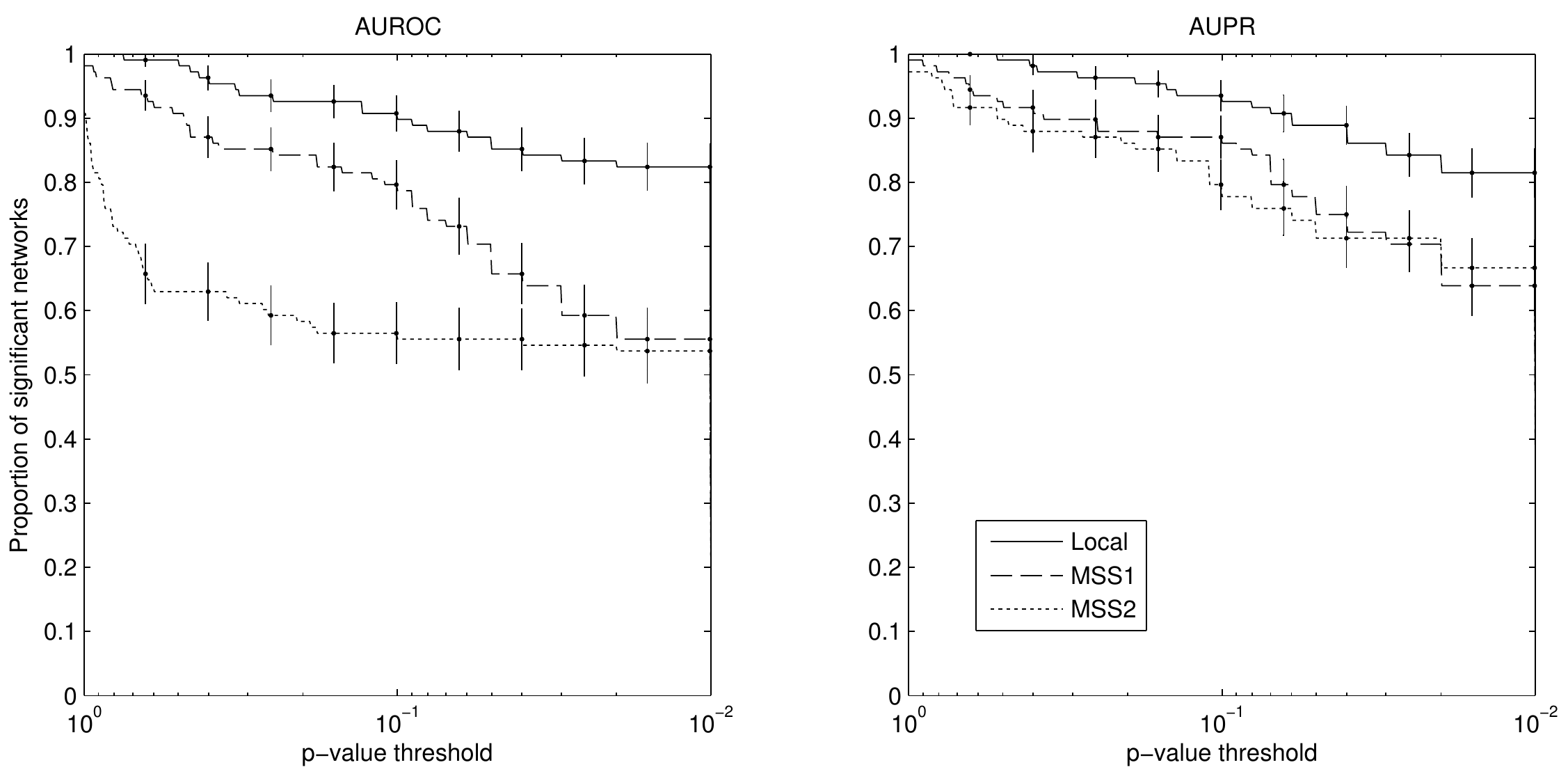}
\caption{Significance testing; empirical study of statistical power.
[Here we considered all 36 DREAM5 methodologies and all 3 datasets, computing Monte Carlo $p$-values for each of the area under the receiver operating characteristic (AUROC) and precision-recall (AUPR) curves, as defined by both local and multi-scale scores. Error bars represent the 68\% confidence region, computed via bootstrap resampling.]}
\label{pvals}
\end{figure}

\section{Discussion} \label{discussion}

Network inference algorithms are increasingly used to facilitate diverse scientific goals, including prediction of single interactions with high precision, identifying motifs such as cliques or clusters, and uncovering global topology such as connected components.
Yet performance assessment of these algorithms does not currently distinguish between these contrasting goals; we do not know which algorithms are most suited to which tasks.
Widely used local scores based on classifier analysis capture the ability of an algorithm to exactly recover a benchmark network at the local level.
In this paper we proposed novel multi-scale scores (MSSs) that instead capture ability to infer connectivity patterns simultaneously over all length scales.
Unlike previous multi-scale proposals, our MSS satisfy desiderata that have contributed to the popularity of local performance scores, including invariance to rank-preserving transformation.
Theoretical results demonstrated the potential for local and multi-scale analyses to draw arbitrarily different conclusions.
Empirical results, based on a comprehensive analysis of the DREAM5 challenge data, confirmed that whilst local and multi-scale performance were correlated, in several cases an estimator recovered local topology well but recovered higher-order topology poorly and {\it vice versa}.
These results highlight the importance of performance assessment on multiple length scales, since apparently promising methodologies may be highly unsuitable for inference of non-local network topology.
Notably the {\it Community} network proposed by \cite{Marbach2012} enjoyed both good local and good multi-scale performance.
Finally we demonstrated that one proposed score ($S_{\text{MSS2}}^{\text{PR}}$) possesses comparable statistical power with respect to local scores.

Our proposals extend scores based on individual edges by considering paths of length $>1$.
Multi-scale analysis of networks is now an established research field; indeed techniques based on \cite{Page} have recently come to light in both gene regulatory network \citep{Morrison,Winter} and protein signalling network \citep{Johannes,Feiglin} analyses.
This work differs to these within-network studies by focussing on the challenge of performance assessment.
The statistical comparison of two unweighted networks has recently been considered by \cite{Stadler,Yates}.
In related work, \cite{Milenkovic} constructed an {\it ad-hoc} null distribution over residue interaction graphs based on degree-distribution, clustering coefficient, clustering spectrum, average network diameter and spectrum of shortest path lengths. However these are {\it mean field} statistics, which do not respect vertex labels, making their use inappropriate for many scientific applications.

Existing and proposed scores capture complementary aspects of performance, so that it is misleading to speak of a universally ``best'' methodology.
MSS was specifically designed to capture aspects of higher-order network reconstruction that have been relatively neglected by the community, however in general the existing and proposed scores together still fail to capture many important aspects of estimator performance. In that respect, this paper represents a small step towards more comprehensive characterisation of estimator performance.
MSS1 required a factor of $\mathcal{O}(p)$ more computation than existing local scores, with MSS2 requiring a factor of $\mathcal{O}(p^2)$. 
(Computational times are presented in Supp. Fig.~\ref{times}.)
Given that in practice it may be desirable to restrict attention to a single measure of multi-scale performance, and given that $S_{\text{MSS1}}^{\text{PR}}$ and $S_{\text{MSS2}}^{\text{PR}}$ provide similar statistical power, it may be preferable to favour MSS1 over MSS2.

Whilst multi-scale analysis has the potential to complement local analyses and deepen our understanding of the applicability and limitations of network inference algorithms, there are two commonly encountered settings where multi-scale analyses may actually be preferred over local analyses:
\begin{enumerate}
\item {\bf Time series data:} \cite{Iwasaki,Dash} and others have noticed that, in the setting of causal graphical models and time series data, the ``true'' causal graph $G$ depends crucially on the time scale at which the process is described.
To see this, notice that the simple multivariate process defined by the causal graph $G$, equal to $\dots \rightarrow X_i \rightarrow X_{i+1} \rightarrow \dots$ and distributions  $X_{i+1}(t+1) = X_i(t) + \epsilon_{i+1}(t+1)$ with $\epsilon_i(t)$ independent $\mathcal{N}(0,\sigma^2)$ random variables, satisfies $X_{i+2}(t+2) = X_i(t) + \tilde{\epsilon}_{i+2}(t+2)$ with $\tilde{\epsilon}$ independent $\mathcal{N}(0,\tilde{\sigma}^2)$ random variables. Thus any consistent estimator of a linear Gaussian vector autoregressive process will, based on data from even time points only, infer the causal graph $\hat{G} = G^2$ equal to $\dots \rightarrow X_i \rightarrow X_{i+2} \rightarrow \dots$. i.e. the apparently natural estimator demonstrates the worst possible performance under $S_{\text{local}}^{\text{ROC/PR}}$.
This is undesirable since it remains the case that all descendancy relationships in $\hat{G}$ are all contained in $G$.
However $\hat{G}$ is readily seen to be $S_{\text{MSS1/MSS2}}^{\text{PR}}$-optimal, so that the proposed scores are robust to the problem of inappropriate sampling frequency.
Real-world examples of this problem arise in the (frequently encountered) settings where no natural time scale is available for experimental design, or the sampling frequency is limited by resource constraints.

\item {\bf Interventional experiments:} In experimental settings, network topology is frequently inferred through targeted interventions (e.g. knock-downs, knock-outs, small molecule inhibitors, monoclonal antibodies etc.). Differential analysis with respect to a control sample in this case does not reveal the direct children of the interventional target, but rather specifies the descendants of that target.
Thus benchmark information is naturally available on descendants, whereas additional experimental work is required to uncover direct edges.
In these settings performance assessment using descendancy (i.e. MSS1) rather than using individual edges facilitates a reduction in experimental cost.
\end{enumerate}

This paper focused on the problem of inference for network structure, but an alternative viewpoint is to assess networks by their predictive power \citep{Prill2}.
In this setting (A4) is replaced by an assumption that estimator weights correspond to effect size (and may be signed). For example these weights may be  average causal effects \citep[ACE;][]{Pearl,Maathuis2}. 
However this setting is less common as many popular network inference methodologies do not entail a predictive statistical model. Moreover, prediction is often computationally intensive, whereas assessment of the inferred topology is relatively cheap.

\section{Conclusion} \label{conclude}

This paper lays statistical and conceptual foundations for the analysis of network inference algorithms on multiple scales.
We restricted attention to the important problem of comparing between inference procedures, as in the DREAM challenges, and proposed novel multi-scale scores (MSSs) for this setting.
An empirical study based on the DREAM5 data demonstrated that multi-scale analysis provides additional insight into the character and capability of network inference algorithms and suggested that a crowd-source approach to inference may offer improved reconstruction of higher-order topology.
In this paper we focussed on connectivity, but multi-scale scores may be designed to capture ability to infer specific motifs such as cliques, or particular feedback circuits, for example.
This work is exploratory and should not be used to form conclusions regarding the performance of specific teams in the DREAM5 challenge.

\section*{Acknowledgement}
The authors are grateful to Gustavo Stolovitzky and Sach Mukherjee for feedback on an earlier draft of this manuscript.
This work was supported by the Centre for Research in Statistical Methodology grant (EPSRC EP/D002060/1). 
CJO was involved in designing the DREAM8 Challenge.
CJO, RA, SEFS were not involved in the DREAM5 Challenge.

\appendix

\section{Proofs}

\begin{proof}[Proof of Theorem \ref{nonunique}]
$S$-optimality of the benchmark network $G$ follows immediately from the definition of each score. 
To prove nonuniqueness we begin by considering MSS1. 
For this, take $G$ to be the network $A \rightarrow B \rightarrow C \rightarrow A$ and take $G'$ to be the network $A \leftarrow B \leftarrow C \leftarrow A$; then both networks entail the same descendancy relationships.
It follows from (D3; consistency) that both $G$ and $G'$ maximise $S_{\text{MSS1}}^{\text{ROC}}$ and $S_{\text{MSS1}}^{\text{PR}}$, demonstrating that the maximiser will be non-unique in general.
The argument for MSS2 is analogous, using the same pair of networks.
\end{proof}

\begin{proof}[Proof of Theorem \ref{different}]
(a) Given $\epsilon > 0$.
We proceed by constructing a sequence of pairs $(G,\hat{G})$ indexed by $p$, the number of vertices, such that the scores associated with $\hat{G}$ end simultaneously to the required limits as $p \rightarrow \infty$.
Define a data-generating network $G$ on $p$ vertices by the edge set 
\begin{eqnarray}
E(G) = \{(1,i) \; : \; 2 \leq i \leq p-1\} \cup \{(i,p) \; : \; 2 \leq i \leq p-1\}
\end{eqnarray}
and consider an unweighted network $\hat{G}$ with topology $E(\hat{G}) = E(G) \cup \{(p,1)\}$. i.e. $\hat{G}$ differs to $G$ in the addition of a single edge $(p,1)$.
The ROC curve corresponding to $\hat{G}$ is defined by the three points 
\begin{eqnarray}
\{(0,0),(\text{FPR},\text{TPR}),(1,1)\}
\end{eqnarray}
so that by linear interpolation the area under the curve is $\frac{1}{2}(1 + \text{TPR} - \text{FPR})$.
The PR curve is similarly defined by 
\begin{eqnarray}
\left\{(0,0),(\text{TPR},\text{PPV}),\left(1,\frac{T}{p(p-1)}\right)\right\}.
\end{eqnarray}
Here $T$ is the number of true examples; for local scores this is $E(G)$, for MSS1 this is $\# \{(i,j) : G(i \rightarrow j) \neq \emptyset, \; i \neq j\}$ and for MSS2 this is $\sum_{i,j} e_{i,j}$.
To compute the area under this curve we must use nonlinear interpolation \citep{Goadrich}.
Specifically, interpolation between the two points $A$, $B$ with true/false positive values $(\text{TP}_A,\text{FP}_A)$, $(\text{TP}_B,\text{FP}_B)$ respectively we create new points for each of $\text{TP}_A + 1, \text{TP}_A + 2, \dots , \text{TP}_B-1$, increasing the false positives for each new point by $(\text{FP}_B - \text{FP}_A)/(\text{TP}_B-\text{TP}_A)$.
By direct calculation we have that the area under this interpolated curve is
\begin{eqnarray}
\frac{\text{TP}^2}{T(\text{TP} + \text{FP})} + \frac{1}{T} \sum_{j=1}^{T - \text{TP}} \frac{(\text{TP}+j)(T - \text{TP})}{(\text{TP} + \text{FP})(T - \text{TP}) + j[p(p-1) = \text{TP} - \text{FP}]}.
\end{eqnarray}

Fix $\epsilon > 0$. 
For (i), since $\hat{G}$ differs to $G$ by just a single edge we have TP $=2(p-2)$, FP $=1$, $T = 2(p-2)$, TPR $= 1$, FPR $= \frac{1}{(p-1)(p-2)}$, PPV $= \frac{2p-4}{2p-3}$. 
It is then easily checked that 
\begin{eqnarray}
S_{\text{local}}^{\text{ROC}} = \frac{1}{2}\left(1 + 1 - \frac{1}{(p-1)(p-2)}\right) \rightarrow 1 \; \; \; \text{as} \; \; \; p \rightarrow \infty, \\
S_{\text{local}}^{\text{PR}} = \frac{[2(p-2)]^2}{[2(p-2)][2(p-2)-1]} \rightarrow 1 \; \; \; \text{as} \; \; \; p \rightarrow \infty.
\end{eqnarray}
Thus $\exists P_1 \in \mathbb{N} : \forall p \geq P_1, S_{\text{local}}^{\text{ROC}} > 1-\epsilon, S_{\text{local}}^{\text{PR}} > 1-\epsilon$.

For (ii-iii) and MSS1 note that $\hat{G}(i \rightarrow j) \neq \emptyset$ for all $i \neq j$, so that $\hat{G}$ predicts every possible descendancy. Thus TP $= 3(p-2)$, FP $=p^2-4p+1$, $T = 3(p-2)$, TPR $= 1$, FPR $= 1$, PPV $= \frac{3(p-2)}{p(p-1)}$ and 
\begin{eqnarray}
S_{\text{MSS1}}^{\text{ROC}} = \frac{1}{2}\left(1 + 1 - 1\right) \rightarrow \frac{1}{2} \; \; \; \text{as} \; \; \; p \rightarrow \infty, \\
S_{\text{MSS1}}^{\text{PR}} = \frac{[3(p-2)]^2}{[3(p-2)][3(p-2)+(p^2-4p+1)]} \rightarrow 0 \; \; \; \text{as} \; \; \; p \rightarrow \infty.
\end{eqnarray}
Thus $\exists P_2 \in \mathbb{N} : \forall p \geq P_2, S_{\text{MSS1}}^{\text{ROC}} < \frac{1}{2}+\epsilon, S_{\text{MSS1}}^{\text{PR}} <\epsilon$.
For MSS2 note that the estimated effects are $\hat{\bm{e}} = \bm{1}_{p \times p}$ whereas the benchmark effects are 
\begin{eqnarray}
\bm{e} = \left[ \begin{array}{ccccc} 1 & 1 & \dots & 1 & 1 \\ & 1 & & & (p-2)^{-1} \\ & & \ddots & & \vdots \\ & & & & (p-2)^{-1} \\ & & & & 1 \end{array} \right]
\end{eqnarray}
leading to TP $= 2p$, FP $= p(p-2)$, $T = 2p$, TPR $= 1$, FPR $= 1$, PPV $= \frac{2}{p}$ and 
\begin{eqnarray}
S_{\text{MSS2}}^{\text{ROC}} = \frac{1}{2}\left(1 + 1 - 1\right) \rightarrow \frac{1}{2} \; \; \; \text{as} \; \; \; p \rightarrow \infty, \\
S_{\text{MSS2}}^{\text{PR}} = \frac{[2p]^2}{[2p][2p+p(p-2)]} \rightarrow 0 \; \; \; \text{as} \; \; \; p \rightarrow \infty.
\end{eqnarray}
Thus $\exists P_3 \in \mathbb{N} : \forall p \geq P_3, S_{\text{MSS2}}^{\text{ROC}} < \frac{1}{2}+\epsilon, S_{\text{MSS2}}^{\text{PR}} <\epsilon$.

Taking $p > \max\{P_1,P_2,P_3\}$, we have shown that $(G,\hat{G})$ satisfy the conclusions of the theorem.

(b) The converse result is proved similarly, taking $G$ to be the cycle $1 \rightarrow 2 \rightarrow \dots \rightarrow p \rightarrow 1$ and $\hat{G}$ to be the reverse  $1 \leftarrow 2 \leftarrow \dots \leftarrow p \leftarrow 1$.
\end{proof}

\begin{proof}[Proof of Theorem \ref{Warshall}]
Let $P=(p_0,\dots,p_m)$ be a path from $p_0$ to $p_m$. Define the intermediate vertices in the path to be the vertices $p_1,\dots,p_{m-1}$. For each pair of vertices $i$ and $j$ we wish to calculate $\tau_{i,j}$, the largest threshold $\tau$ for which $i$ and $j$ are connected by a path in $H^{\tau}$. 

For $k=0,\dots,p$, let $\tau_{i,j}^{(k)}$ represent the largest value of the threshold $\tau$ for which $H^{\tau}$ contains a path from $i$ to $j$ with intermediate vertices in the set $\{1,\dots,k\}$. From this definition, we have $\tau_{i,j}^{(0)}=H$, the weighted adjacency matrix and  $\tau_{i,j}=\tau_{i,j}^{(p)}$. Next we show that for $k=1,\dots,p$, 
\begin{eqnarray}
\tau_{i,j}^{(k)}=\max\{\tau_{i,j}^{(k-1)},\min\{\tau_{i,k}^{(k-1)},\tau_{k,j}^{(k-1)}\}\}. 
\label{recurse}
\end{eqnarray}
To see this, note that $\tau_{i,j}^{(k)}=\tau_{i,j}^{(k-1)}$ unless there is a more strongly connected path which goes through $k$. Such a path must first connect $i$ to $k$, and then connect $k$ to $j$, with both paths involving intermediate vertices in $\{1,\dots,k-1\}$. This connection is broken in $H^{\tau}$ once the threshold $\tau$ is raised above the minimum of these two connections, that is when $\tau > \min\{\tau_{i,k}^{(k-1)},\tau_{k,j}^{(k-1)}\}$.

The conclusion of the Theorem follows by iterative application of Eqn.~\ref{recurse}.
\end{proof}

\noindent {\it Remarks on the proof of Theorem \ref{Warshall}:} To make the algorithm more memory efficient, we note that it is not necessary to store more than $\mathcal{O}(p^2)$ values for $\tau$ at any point. This stems from the fact that each iteration is calculated only from the results of the previous one, reducing the necessary storage to $2p^2$. This can be improved further by noting that if $i=k$ or $j=k$ then $k$ cannot be an intermediate vertex in a path connecting $i$ to $j$ and hence $\tau_{i,j}^{(k)}=\tau_{i,j}^{(k-1)}$. These $2k-1$ index values are the only elements needed to calculate the $k$th iteration, and they themselves do not need to be updated. Hence, the remaining elements in the matrix are needed only for their own update and can be safely over-written. A corollary to this is that the $i$ and $j$ loops in the algorithm can be completed in any order or even in parallel.
Note that this algorithm works whether or not the weight matrix $H$ is allowed to contain loops. Irrespective of whether loops are permitted or not, $\tau_{i,i}$ represents the largest value of the edge threshold for which node $i$ is contained in a cycle.

For Theorem \ref{compute} we require a series of Lemmas:

\begin{lem}
\label{converge}
The effects 
\begin{eqnarray}
e_{ij} = \delta_{ij} + \sum_{P\in\mathcal{P}(i,j)} \prod_{k=1}^{\ell(P)} \frac{1}{|H(\bullet,p_k)|}
\label{effects}
\end{eqnarray}
are well-defined and satisfy $0 \leq e_{ij} \leq 1$.
\end{lem}
\begin{proof}
We show the sum in Eqn.~\ref{effects} converges absolutely.
Since all terms are non-negative, it is sufficient to prove that, for each $i \neq j$ such that $\mathcal{P}(i,j) \neq \emptyset$, the sum is bounded above.
Consider the partial sums
\begin{eqnarray}
e_{ij}^{(n)} = \sum_{P \in \mathcal{P}(i,j) : \ell(P) \leq n} \prod_{k=1}^{\ell(P)} \frac{1}{|H(\bullet,p_k)|}.
\end{eqnarray}
If Eqn.~\ref{effects} diverges, so must $\lim_{n \rightarrow \infty} e_{ij}^{(n)}$ since all paths are of finite length, so it is sufficient to prove $e_{ij}^{(n)}$ is bounded above by one as $n \rightarrow \infty$.
We proceed inductively, with base case 
\begin{eqnarray}
e_{ij}^{(1)} = \left\{
     \begin{array}{lr}
       \frac{1}{|H(\bullet,j)|} & : i \in H(\bullet,j) \\
       0 & : \text{otherwise}
     \end{array}
   \right\}  \leq 1.
\end{eqnarray}
Suppose $e_{ij}^{(n-1)} \leq 1$. Then for $i \notin H(\bullet,j)$ we have
\begin{eqnarray}
e_{ij}^{(n)} = \sum_{k \in H(\bullet,j)} \frac{1}{|H(\bullet,j)|} e_{ik}^{(n-1)}
\end{eqnarray}
since a path $P$ of length $m$ can be decomposed into a path from $p_0 = i$ to $p_{m-1} = k$ and an edge from $p_{m-1} = k$ to $p_m = j$.
Therefore $e_{ij}^{(n)} \leq \sum_{k \in H(\bullet,j)} \frac{1}{|H(\bullet,j)|} = 1$ as required.
The case where $i \in H(\bullet,j)$ is similar, leading to
\begin{eqnarray}
e_{ij}^{(n)} = \sum_{k \in H(\bullet,j) \setminus \{i\}} \frac{1}{|H(\bullet,j)|} e_{ik}^{(n-1)}
\end{eqnarray}
from which we again conclude $e_{ij}^{(n)} \leq \sum_{k \in H(\bullet,j)} \frac{1}{|H(\bullet,j)|} = 1$.
Therefore $e_{ij}^{(n)} \leq 1$ for all $n$ by induction.
\end{proof}

\begin{lem}
\label{equivalent}
Effects, as defined in Eqn.~\ref{effects}, satisfy (i) $e_{ij} = \delta_{ij}$ for all $i,j$ s.t. $\mathcal{P}(i,j) = \emptyset$, (ii) $e_{ij} = \sum_{k \in H(\bullet,j)} \frac{1}{|H(\bullet,j)|} e_{ik}$ for all $i,j$ s.t. $\mathcal{P}(i,j) \neq \emptyset$.
\end{lem}
\begin{proof}[Proof of Lemma \ref{equivalent}]
The first statement follows immediately from Eqn.~\ref{effects}. For the second statement, consider two cases: Firstly, for $i \neq j$ such that $i \notin H(\bullet,j)$ we have 
\begin{eqnarray}
\sum_{k \in H(\bullet,j)} \frac{1}{|H(\bullet,j)|} e_{ik} & = & \sum_{k \in H(\bullet,j)} \frac{1}{|H(\bullet,j)|} \sum_{P\in\mathcal{P}(i,k)} \prod_{l=1}^{\ell(P)} \frac{1}{|H(\bullet,p_l)|}  \\
& = & \sum_{P \in \mathcal{P}(i,j)} \prod_{l=1}^{\ell (P)} \frac{1}{|H(\bullet,p_l)|} = e_{ij}.
\end{eqnarray}
Note that absolute convergence (Lemma \ref{converge}) ensures the manipulation of infinite sums is valid.
The proof for $i \in H(\bullet,j)$ is similar but requires slightly more care:
\begin{eqnarray}
\sum_{k \in H(\bullet,j)} \frac{1}{|H(\bullet,j)|} e_{ik} & = & \sum_{k \in H(\bullet,j)} \frac{1}{|H(\bullet,j)|} \left\{ \delta_{ik} + \sum_{P\in\mathcal{P}(i,k)} \prod_{l=1}^{\ell(P)} \frac{1}{|H(\bullet,p_l)|} \right\} \\
& = & \frac{1}{|H(\bullet,j)|} + \sum_{k \neq i} \frac{1}{|H(\bullet,j)|} \sum_{P \in \mathcal{P}(i,k)} \prod_{l=1}^{\ell(P)} \frac{1}{|H(\bullet,p_l)|}  \\
& = & \sum_{\substack{P \in \mathcal{P}(i,j) \\ p_{\ell(P) - 1} = i}} \prod_{l=1}^{\ell (P)} \frac{1}{|H(\bullet,p_l)|} + \sum_{\substack{P \in \mathcal{P}(i,j) \\ p_{\ell(P) - 1} \neq i}} \prod_{l=1}^{\ell (P)} \frac{1}{|H(\bullet,p_l)|}  \\
& = & \sum_{P \in \mathcal{P}(i,j)} \prod_{l=1}^{\ell (P)} \frac{1}{|H(\bullet,p_l)|} = e_{ij}.
\end{eqnarray}
\end{proof}

\begin{lem}
\label{character}
Define a matrix $\bm{M}^{(i)}$ by (i) $M_{im}^{(i)} = \delta_{im}$, (ii) for $\mathcal{P}(i,k) \neq \emptyset$ and $m \in H(\bullet,k)$, set $M_{km}^{(i)} = \frac{1}{|H(\bullet,k)|}$, (iii) $M_{km}^{(i)} = 0$ otherwise.
Then the conclusions of Lemma \ref{equivalent} are equivalent to the following:
For each $i$, the vector $\bm{v}$ where $v_k = e_{ik}$ satisfies $\bm{M}^{(i)}\bm{v} = \bm{v}$.
\end{lem}
\begin{proof}
Equivalence may be verified directly.
\end{proof}

\begin{lem}
Suppose $\bm{v}$ satisfies $\bm{M}^{(i)}\bm{v} = \bm{v}$. Then for all $k$ we must have $|v_k| \leq |v_i|$.
\label{nonzero}
\end{lem}
\begin{proof}
Suppose not. Then there exists $k$ such that (i) $|v_k|>|v_i|$ and (ii) $|v_k|$ is maximal.
Without loss of generality $\mathcal{P}(i,k) \neq \emptyset$ since otherwise $v_k = \sum_j M_{kj}^{(i)}v_j = \sum_j 0 v_j= 0$.
Then $|v_m| = |v_k|$ for all $m \in H(\bullet,k)$, since $v_k = \sum_m M_{km}^{(i)}v_m = \sum_{m \in H(\bullet,k)} \frac{1}{|H(\bullet,k)|} v_m$ and $\sum_{m \in H(\bullet,k)} \frac{1}{|H(\bullet,k)|}  = 1$.
By induction on a path from $i$ to $k$ we arrive at $|v_i| = |v_k|$, contradicting (i).
\end{proof}

\begin{lem}
There exists a unique solution to $\bm{M}^{(i)}\bm{v} = \bm{v}$ with $v_i = 1$.
\end{lem}
\begin{proof}
{\it Existence:}
The $i$th row of $\bm{M}^{(i)}$ contains only zeros, except for the $i$th entry which is equal to 1. Hence the characteristic polynomial contains a root (eigenvalue) $\lambda = 1$ with corresponding eigenvector $\tilde{\bm{v}} \neq \bm{0}$. From Lemma \ref{nonzero} we must have $v_{i} \neq 0$ (else $|\tilde{v}_j| \leq |\tilde{v}_i| = 0$ for all $j$, implying that $\tilde{\bm{v}} = \bm{0}$). Taking $\bm{v} = \tilde{\bm{v}} / v_{ii}$ proves existence.

{\it Uniqueness:} Suppose $\bm{v}$ and $\bm{v}'$ both satisfy the criteria. Then $\bm{v}'' = \bm{v}-\bm{v}'$ is also an eigenvector of $\bm{M}^{(i)}$ with unit eigenvalue. From Lemma \ref{nonzero} we have $|v_k''| \leq |v_i''| = 0$ for all $k$, implying that $\bm{v} = \bm{v}'$.
\end{proof}

\begin{proof}[Proof of Theorem \ref{compute}]
Together Lemmas 1-5 prove that the effect matrix $\bm{e} = \{e_{ij}\}$ exists and is unique. Moreover, the $i$th row $\bm{e}_i$ of $\bm{e}$ may be computed as the unique eigenvector of $\bm{M}^{(i)}$ corresponding to unit eigenvalue.

For moderate dimensional matrices $\bm{M}^{(i)}$, eigenvectors may be calculated using any valid algorithm. One such algorithm, suitable for high dimensional problems, is the ``power iteration'' method. Specifically we compute the leading eigenvector $\bm{v}$ of $\bm{M}^{(i)}$ as
\begin{eqnarray}
\bm{e}_i = \lim_{n \rightarrow \infty} \bm{v}^{(n)}, \; \; \; \bm{v}^{(n)} = \frac{\bm{M}^{(i)}\bm{v}^{(n-1)}}{\|\bm{M}^{(i)}\bm{v}^{(n-1)}\|}, \; \; \; v^{(0)}_k = \delta_{ik}.
\end{eqnarray}
Efficient computation in this setting is reviewed in \cite{Berkhin}.
\end{proof}

\newpage
\pagenumbering{gobble}
\section{Supplement to ``Quantifying the Multi-Scale Performance of Network Inference Algorithms''}

\subsection{AUC Convention} \label{AUC con}
For MSS2 only, the analogue of the confusion matrix does not respect monotonicity of the entries TP, FP, TN, FN with respect to varying threshold parameter $\tau$. As a consequence both the ROC and PR curves constructed by interpolation have the potential to exhibit self-intersection and hence there may exist estimates $\hat{G}$ such that the area $S$ under the associated ROC and PR curves (AUC) need not be bounded by $0 \leq S \leq 1$, invalidating the definition of a score function having image in $[0,1]$.
To construct a consistent performance measure we instead define AUC as follows:
Let $(x_1,y_1) \rightarrow (x_2,y_2) \rightarrow \dots \rightarrow (x_m,y_m)$ denote the coordinates of points on the curve whose AUC is to be calculated. Reorder the sequence as $(x_{n_1},y_{n_1}) \rightarrow (x_{n_2},y_{n_2}) \rightarrow \dots \rightarrow (x_{n_m},y_{n_m})$ such that the x-coordinates form an increasing sequence $x_{n_1} \leq x_{n_2} \leq \dots \leq x_{n_m}$.
Then AUC is defined using the trapezium rule as $S = \sum_{i = 2}^m \frac{y_{n_{i-1}}+y_{n_i}}{2} \times (x_{n_i}-x_{n_{i-1}})$.
Note that whilst an ordering $\{x_{n_j}\}$ may in general be non-unique, the expression for $S$ does not depend on which particular ordering is selected.

\subsection{The ``net assess'' Package}
The \verb+net_assess+ package for MATLAB R2013b is provided to compute all of the scores and associated $p$-values discussed in this paper.
In the case of MSS2 we compute ROC and PR curves based on 100 uniformly spaced thresholds $\tau_i$ defined implicitly by $|E(\hat{G}^{\tau_i})| = (i/100) \times |E(\hat{G})|$.
This package can be used as shown in the following example:

\verb+A2 = rand(100)>0.98; \% benchmark network+

\verb+A1 = 0.5*rand(100)+ + \verb+0.5*rand(100).*A2; \% estimate+

\verb+score = 1; \% 1 = local, 2 = MSS1, 3 = MSS2+

\verb+[AUROC,AUPR] = net_assess(score,A1,A2)+

\verb+MC_its = 100; \% num. Monte-Carlo samples for p-val. computation+

\verb+[AUROC,AUPR,p_AUROC,p_AUPR] = net_assess(score,A1,A2,MC_its)+

Typical computational times for the package, based on the DREAM5 data analysis performed in this paper, are displayed in Supp. Fig.~\ref{times}.
Note that local scores require $\mathcal{O}(p^2)$ operations to evaluate, whereas MSS1 requires $\mathcal{O}(p^3)$ and MSS2 requires $\mathcal{O}(p^4)$.

\subsection{Overall Score (Marbach {\it et al.})} \label{overall def}

We briefly summarise the approach of \cite{Marbach2012} that assigns each network inference algorithm an ``Overall Score''.
For each network inference algorithm a combined AUROC score was calculated as the (negative logarithm of the) geometric mean over $p$-values corresponding to 3 individual AUROC scores ({\it in silico}, {\it E. Coli}, {\it S. Cerevisiae} respectively). Likewise a combined AUPR score was assigned to each network inference algorithm. For each algorithm an Overall Score was computed as the arithmetic mean of the combined AUROC and combined AUPR scores.
Note that \cite{Marbach2012} report the Overall Score over all 3 datasets, whereas the performance scores which we report in this paper were computed for each dataset individually.
Note also that $p$-values used by \cite{Marbach2012} were computed using a common null distribution, rather than the estimator-specific nulls $\mathcal{M}_0(\hat{G})$ that are used in this paper.

\subsection{Additional Results} \label{AUROC results}

Supp. Figs. \ref{silico bar}, \ref{ecoli bar}, \ref{sc bar} display the scores $S$ obtained by each of the 36 dream methodologies on, respectively, the {\it in silico}, {\it E. Coli} and {\it S. Cerevisiae} datasets.
Supp. Fig.~\ref{DREAM2} compares results for the DREAM5 Challenge data based on both local and multi-scale scores and ROC analysis. Supp. Fig.~\ref{meta5} displays both inferred benchmark topology for the {\it in silico} dataset using {\it Regression 2}.

\newpage

\begin{figure}
\centering
\includegraphics[width = \textwidth,clip,trim = 4cm 0cm 0cm 1cm]{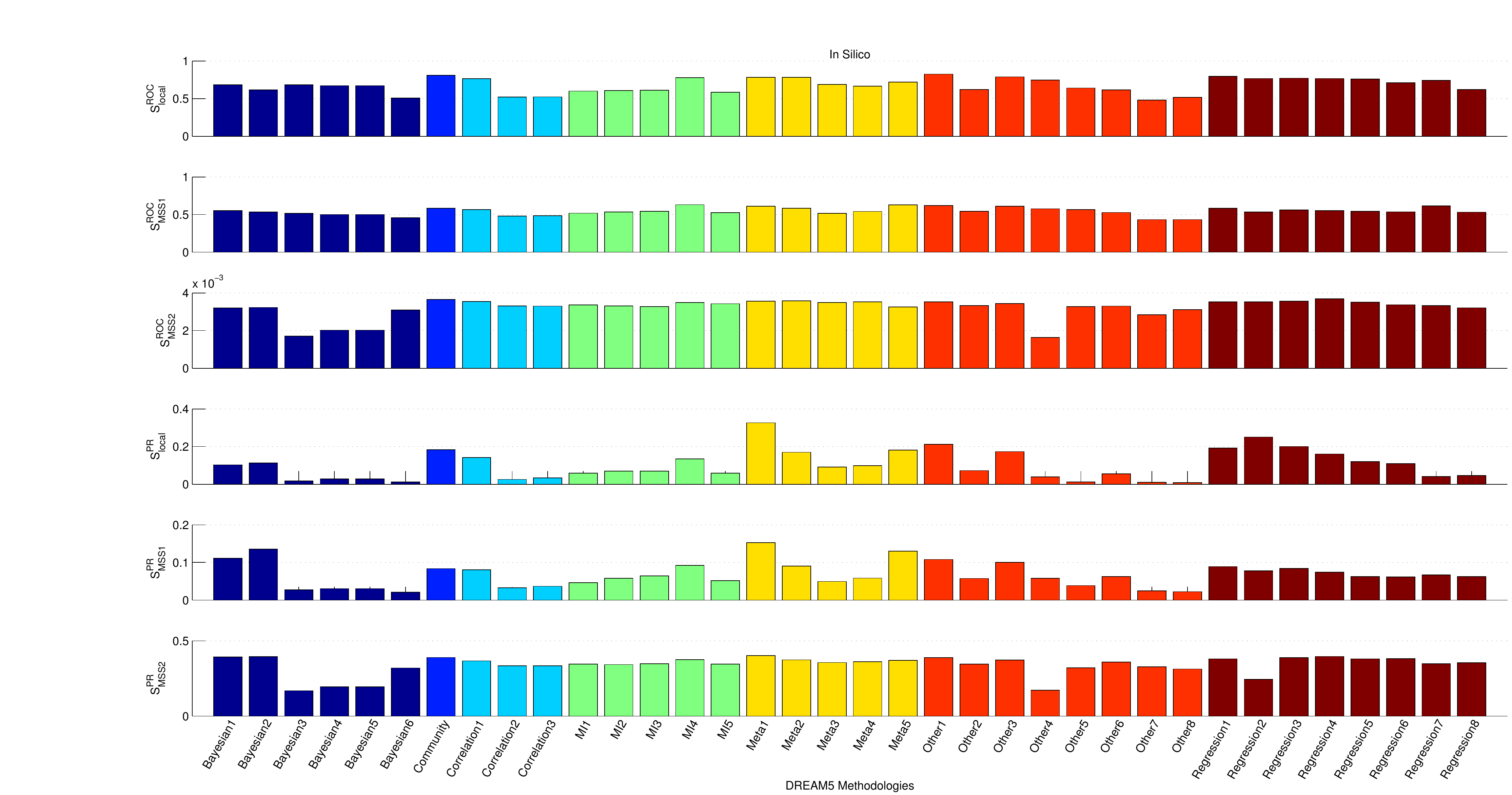}
\caption{DREAM5 Challenge data; results on the {\it in silico} dataset.
 [The 36 DREAM5 methodologies were assessed using both local and multi-scale performance scores. Participants were grouped according to their statistical formulation as either {\it Bayesian}, {\it Correlation}, {\it Meta}, {\it Mutual Information} (MI), {\it Regression} or {\it Other}.
{\it Community} represents a crowd-sourced network estimator proposed by \cite{Marbach2012}.
Performance scores include area under the receiver operating characteristic (AUROC) and precision recall (AUPR) curves, based on both local scores and the proposed multi-scale scores (MSS). The Overall Score of \cite{Marbach2012} combines both local AUROC and AUPR scores as described in Supp. Sec.~\ref{overall def}.]}
\label{silico bar}
\end{figure}

\begin{figure}
\centering
\includegraphics[width = \textwidth,clip,trim = 4cm 0cm 0cm 1cm]{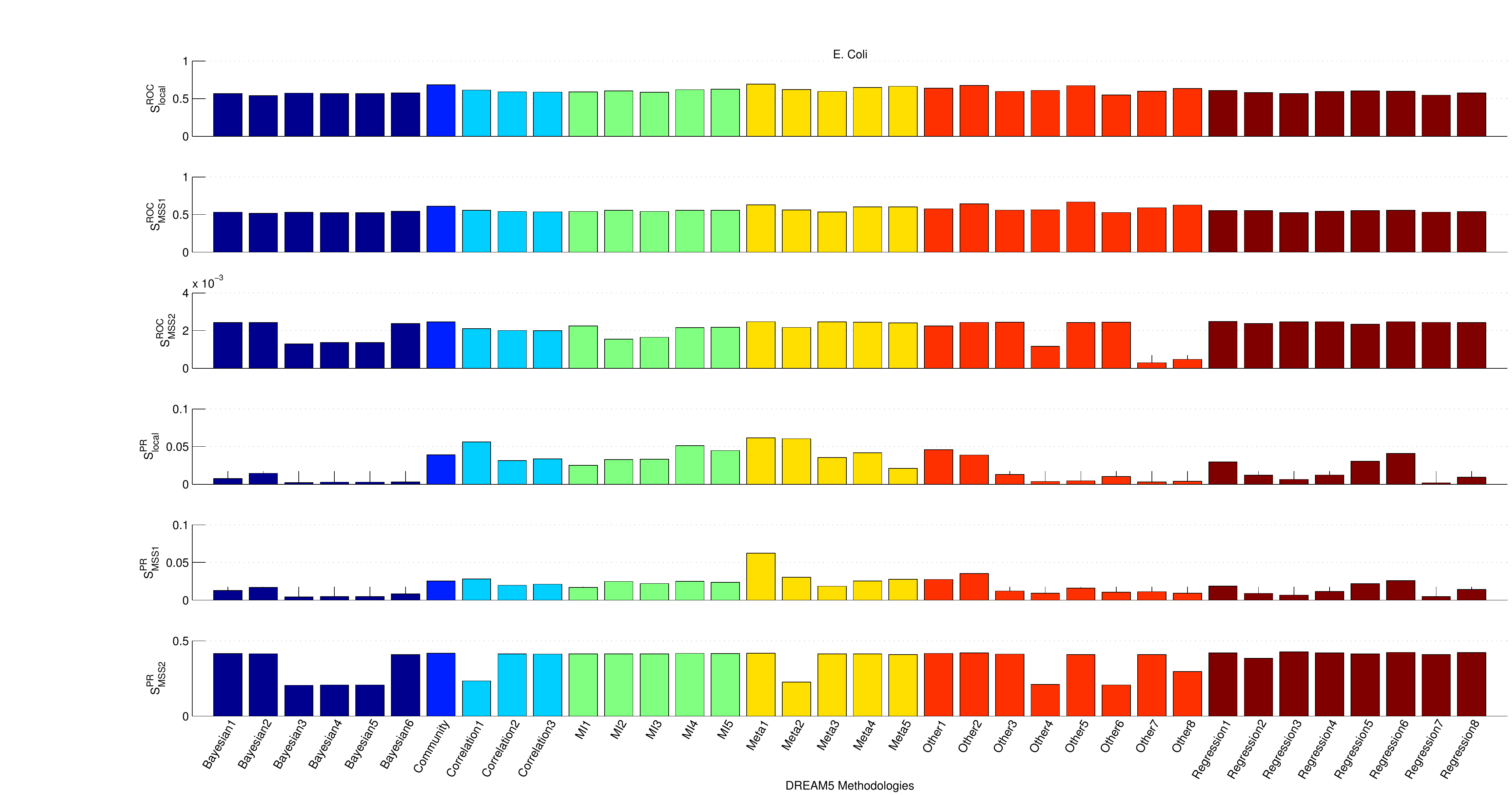}
\caption{DREAM5 Challenge data; results on the {\it E. Coli} dataset.
 [The 36 DREAM5 methodologies were assessed using both local and multi-scale performance scores. Participants were grouped according to their statistical formulation as either {\it Bayesian}, {\it Correlation}, {\it Meta}, {\it Mutual Information} (MI), {\it Regression} or {\it Other}.
{\it Community} represents a crowd-sourced network estimator proposed by \cite{Marbach2012}.
Performance scores include area under the receiver operating characteristic (AUROC) and precision recall (AUPR) curves, based on both local scores and the proposed multi-scale scores (MSS). The Overall Score of \cite{Marbach2012} combines both local AUROC and AUPR scores as described in Supp. Sec.~\ref{overall def}.]}
\label{ecoli bar}
\end{figure}

\begin{figure}
\centering
\includegraphics[width = \textwidth,clip,trim = 4cm 0cm 0cm 1cm]{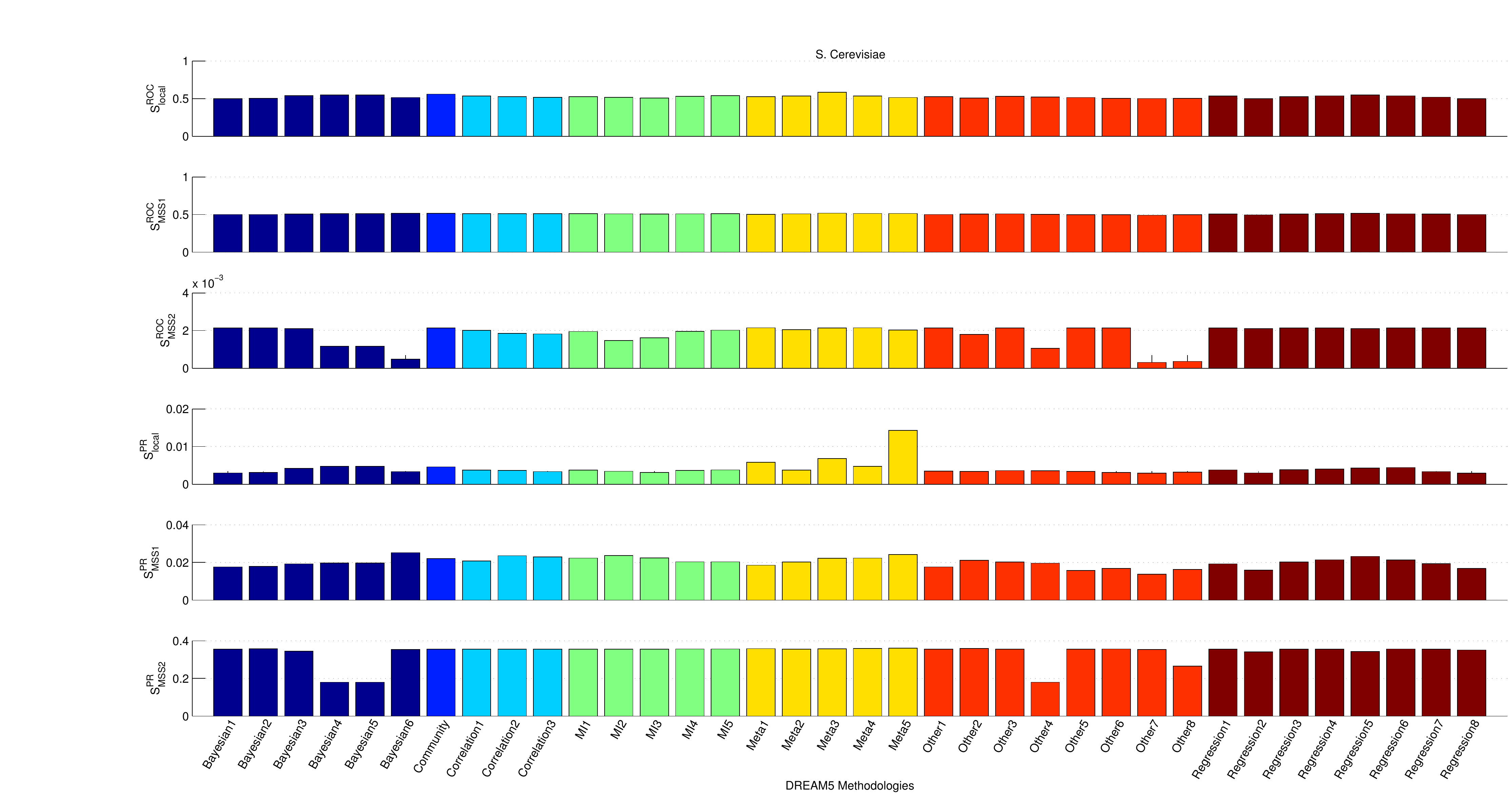}
\caption{DREAM5 Challenge data; results on the {\it S. Cerevisiae} dataset.
 [The 36 DREAM5 methodologies were assessed using both local and multi-scale performance scores. Participants were grouped according to their statistical formulation as either {\it Bayesian}, {\it Correlation}, {\it Meta}, {\it Mutual Information} (MI), {\it Regression} or {\it Other}.
{\it Community} represents a crowd-sourced network estimator proposed by \cite{Marbach2012}.
Performance scores include area under the receiver operating characteristic (AUROC) and precision recall (AUPR) curves, based on both local scores and the proposed multi-scale scores (MSS). The Overall Score of \cite{Marbach2012} combines both local AUROC and AUPR scores as described in Supp. Sec.~\ref{overall def}.]}
\label{sc bar}
\end{figure}

\begin{figure}
\centering
\includegraphics[width = \textwidth]{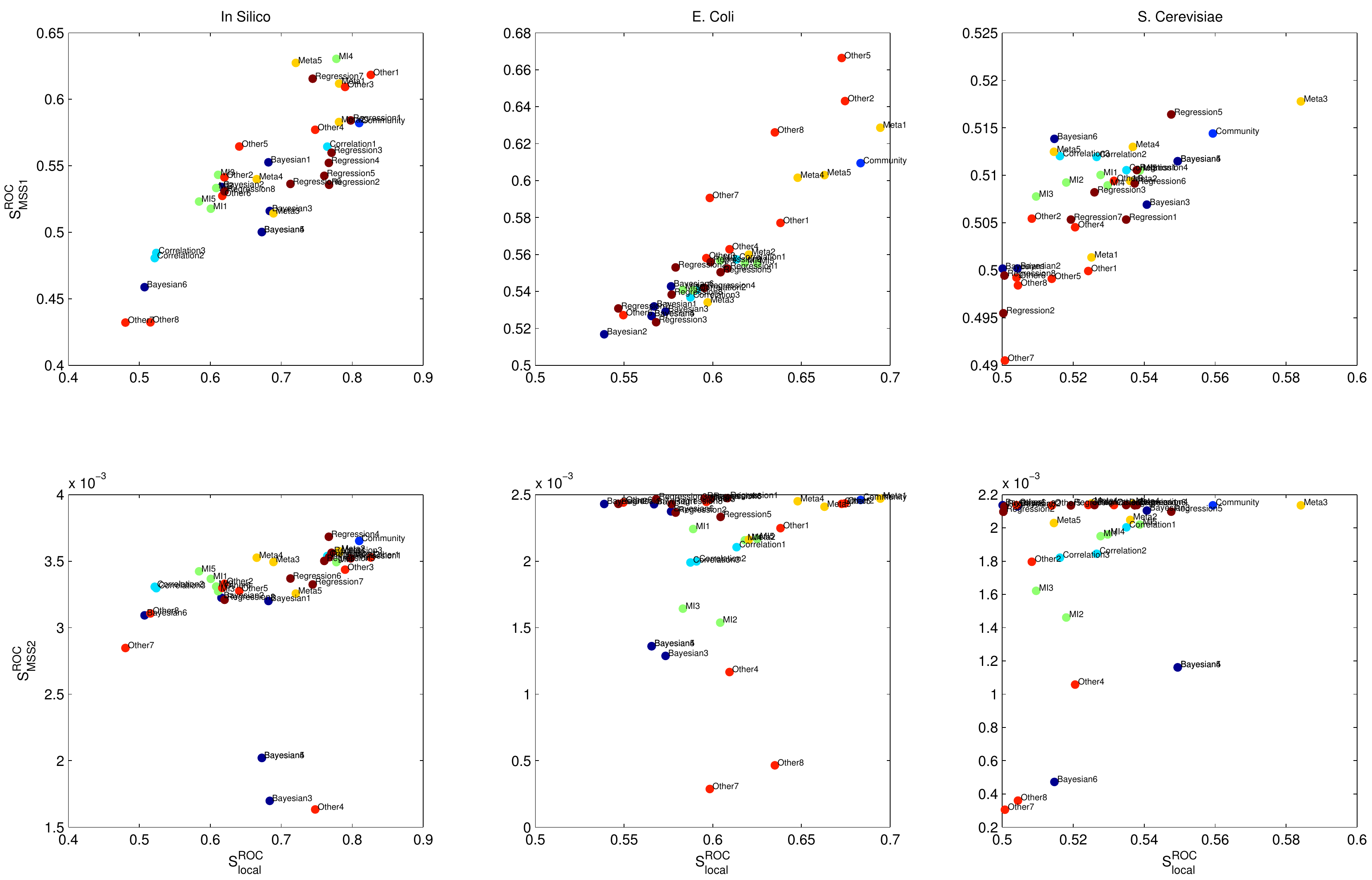}
\caption{DREAM5 network inference challenge data. [The 36 DREAM5 methodologies were assessed using both local and multi-scale performance scores. Participants were grouped according to their statistical formulation as either {\it Bayesian}, {\it Correlation}, {\it Meta}, {\it Mutual Information} (MI), {\it Regression} or {\it Other}.
{\it Community} represents a crowd-sourced network estimator proposed by \cite{Marbach2012}.
Left to right: {\it In silico} dataset, {\it E. Coli} dataset, {\it S. Cerevisiae} dataset.
Here we show results based on area under the receiver operating characteristic curve.]}
\label{DREAM2}
\end{figure}

\begin{figure}
\centering
\begin{subfigure}{\textwidth}
\pgfdeclarelayer{foreground}
\pgfdeclarelayer{background}
\pgfsetlayers{background,main,foreground}
\begin{tikzpicture}
\begin{pgfonlayer}{foreground}
\draw [red] (5.4,0.8) rectangle (8.4,2);
\end{pgfonlayer}
\begin{pgfonlayer}{background}
\centering
\path (0,0) node (o) {
\includegraphics[width = \textwidth]{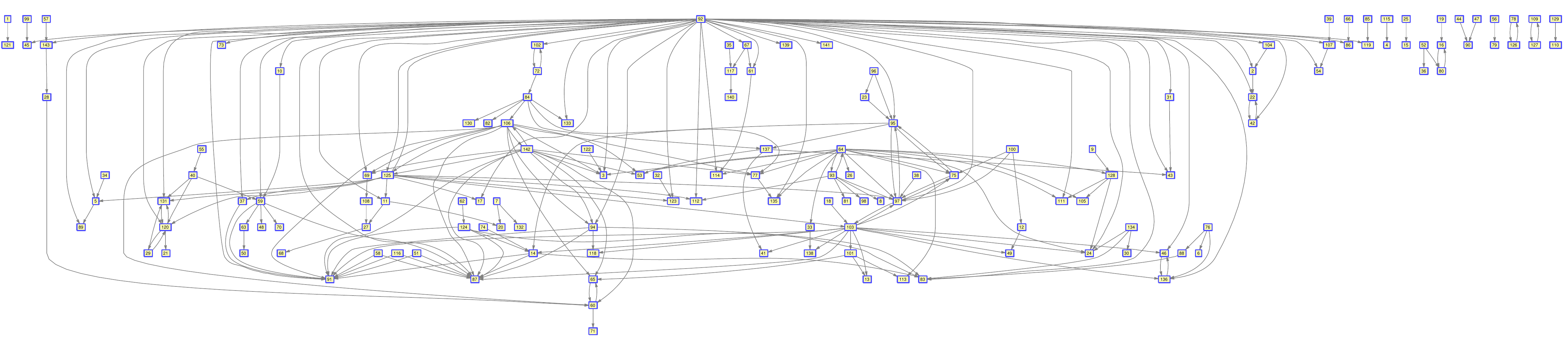}};
\end{pgfonlayer}
\end{tikzpicture}
\caption{{\it In silico} data-generating network}
\end{subfigure}
\begin{subfigure}{\textwidth}
\pgfdeclarelayer{foreground}
\pgfdeclarelayer{background}
\pgfsetlayers{background,main,foreground}
\begin{tikzpicture}
\begin{pgfonlayer}{foreground}
\draw [red] (5.4,0.7) rectangle (8.4,1.9);
\end{pgfonlayer}
\begin{pgfonlayer}{background}
\centering
\path (0,0) node (o) {
\includegraphics[width = \textwidth]{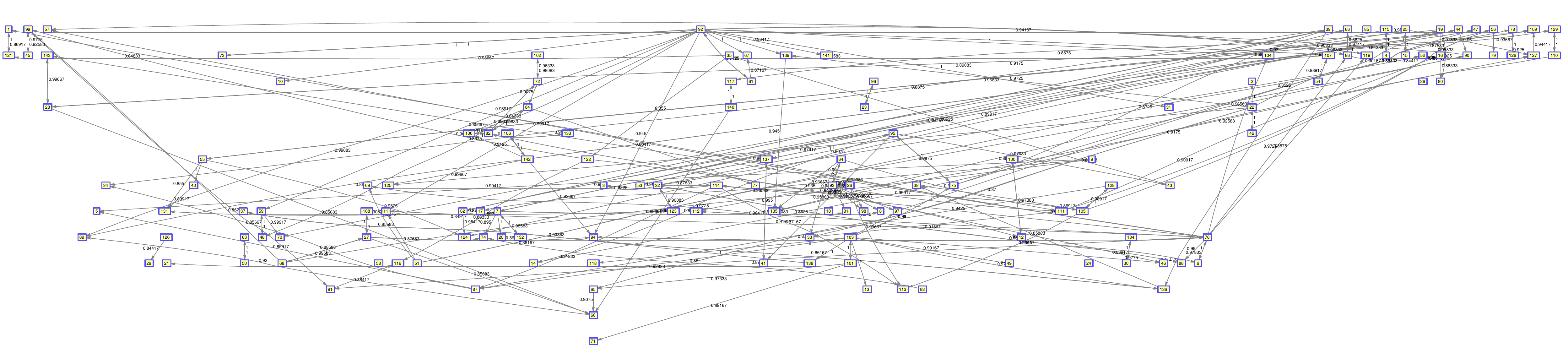}};
\end{pgfonlayer}
\end{tikzpicture}
\caption{Network topology inferred by {\it Regression2}}
\end{subfigure}
\caption{Global connectivity patterns in the DREAM5 {\it in silico} chellenge; {\it Regression2} is unable to identify disconnected components (red box) in the data-generating network. [High definition available in electronic version.] (a) {\it In silico} data-generating network, based on known transcriptional regulatory networks of model organisms \citep{Marbach2009}. (b) Network topology inferred by {\it Regression2}; an approach that combines steady-state and time-series data using a group LASSO, followed by bootstrapping.
[Edges and associated weights are shown for the highest weights only, such that both networks (a) and (b) contain an equal number of edges.]}
\label{meta5}
\end{figure}

\begin{figure}
\centering
\includegraphics[width = 0.5\textwidth]{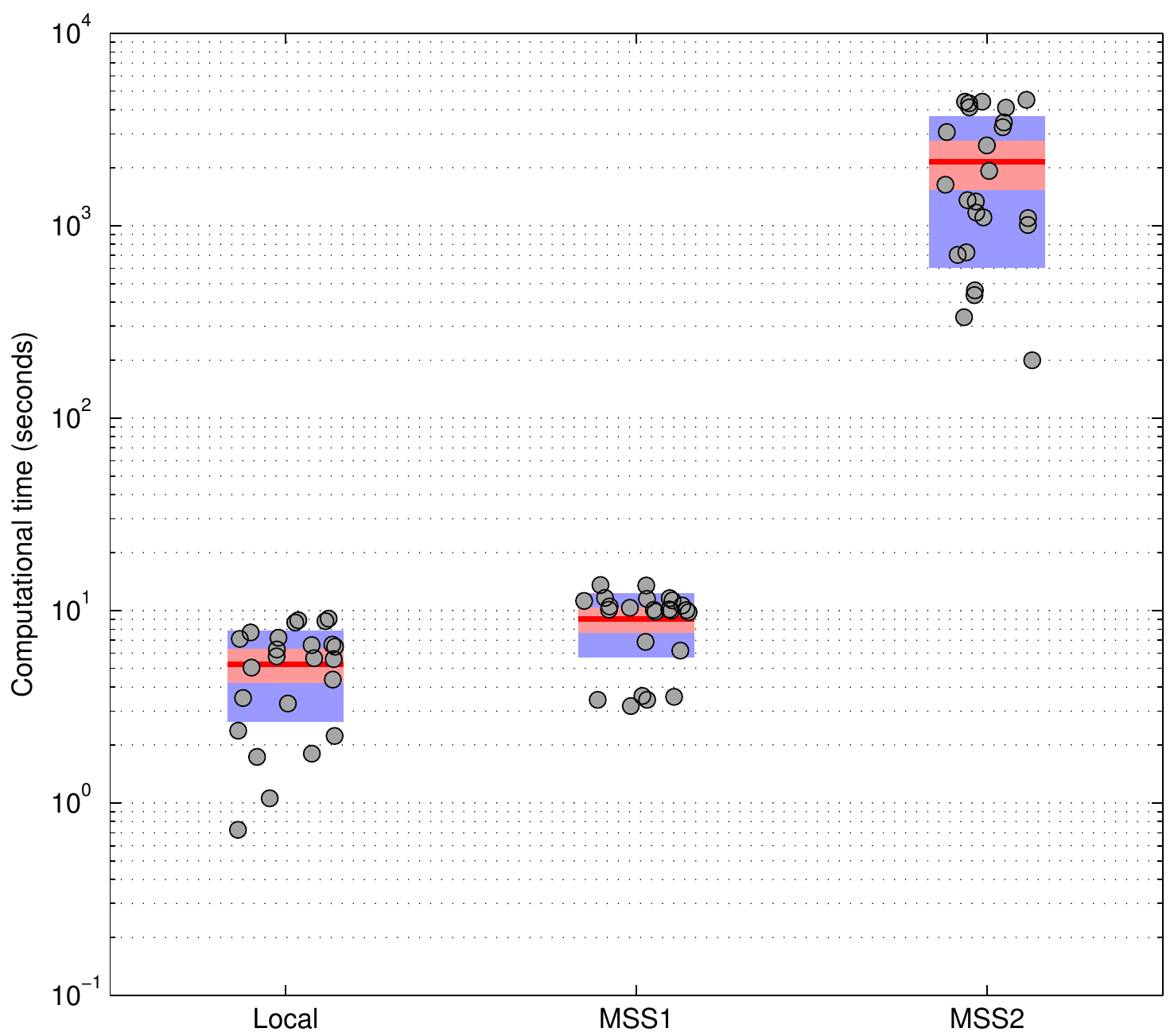}
\caption{Computational times for the DREAM5 Challenge data. [Red; 95\% confidence intervals for the mean. Blue; inter-quartile range. The full data are shown as points.]}
\label{times}
\end{figure}

\end{document}